\documentclass[conference]{IEEEtran}
\IEEEoverridecommandlockouts

\usepackage{graphicx}
\usepackage{mathptmx}
\usepackage{amsmath}
\usepackage{amssymb}
\usepackage{amsthm}
\usepackage{booktabs}   
\usepackage{subcaption} 
\usepackage{algorithmic}
\usepackage[linesnumbered,ruled]{algorithm2e}
\usepackage{wrapfig}
\usepackage{bcprules}

\usepackage{bussproofs}
\usepackage{proof}
\usepackage{amssymb}

\newcommand{\tool}{\textsc{Remedy}}
\newcommand{\distfunc}{{\it D}}
\newcommand{\average}[1]{\ensuremath{\langle#1\rangle} }

\newcommand{\queue}{{\it Q}}
\newcommand{\idt}{\hspace{0.5cm}}

\newcommand{\hole}{\Box}
\newcommand{\sstate}{{\it t}}
\newcommand{\any}{\cdot}

\newcommand{\encode}{\dashrightarrow}

\newcommand{\trace}{{\it tr}}
\newcommand{\run}{{\sc Run}}
\newcommand{\computePath}{{\sc Bps}}
\newcommand{\follow}{\textit{Fst}}
\newcommand{\Bracketing}{\Psi}

\newcommand{\toENFA}{{\it eNFAtr}}

\newcommand{\runtime}{\mathit{Time}}
\newcommand{\ltp}{RWS1U}
\newcommand{\removela}{\textit{rmla}}
\newcommand{\tapprox}{\alpha}

\newcommand{\supplementary}{Appendix} 
\newtheorem{definition}{Definition}[section]
\newtheorem{theorem}{Theorem}[section]
\newtheorem{lemma}{Lemma}[section]

\newtheorem{example}{Example}[section]
\usepackage[htt]{hyphenat}

\usepackage[normalem]{ulem}
\newcommand{\stkout}[1]{\ifmmode\text{\sout{\ensuremath{#1}}}\else\sout{#1}\fi}

\newcommand{\tchanged}[2]{#2}

\usepackage{cite}
\usepackage{amsmath,amssymb,amsfonts}
\usepackage{algorithmic}
\usepackage{graphicx}
\usepackage{textcomp}
\usepackage{xcolor}
\def\BibTeX{{\rm B\kern-.05em{\sc i\kern-.025em b}\kern-.08em
    T\kern-.1667em\lower.7ex\hbox{E}\kern-.125emX}}

\usepackage{colortbl}
\usepackage{bbding}
\usepackage{pifont}
\usepackage{wasysym}
\newcommand{\xmark}{\ding{55}}
\usepackage{booktabs}
\newcommand{\togform}{{\mathit AtoG}}
\newcommand{\leaf}{{\sf leaf}}
\newcommand{\groot}{{\sf root}}
\newcommand{\indig}{d_{in}}
\newcommand{\outdig}{d_{out}}

\newcommand{\tail}{{\sf tail}}
\newcommand{\head}{{\sf head}}
\newcommand{\path}{\mathfrak{p}}
\newcommand{\freshi}{\mathfrak{i}}
\newcommand{\fresh}{\mathit{id}}
\newcommand{\texcomment}[1]{}

\begin{document}

\title{Repairing DoS Vulnerability of Real-World Regexes}


\author{\IEEEauthorblockN{Nariyoshi Chida}
\IEEEauthorblockA{
\textit{NTT Corporation / Waseda University}\\
nariyoshichidamm@gmail.com}
\and
\IEEEauthorblockN{Tachio Terauchi}
\IEEEauthorblockA{
\textit{Waseda University}\\
terauchi@waseda.jp}
}

\maketitle

\pagestyle{plain}

\begin{abstract}
There has been much work on synthesizing and repairing regular expressions ({\em regexes} for short) from examples.  These {\em programming-by-example} (PBE) methods help the users write regexes by letting them reflect their intention by examples.  However, the existing methods may generate regexes whose matching may take super-linear time and are vulnerable to regex denial of service (ReDoS) attacks.  This paper presents the {\em first} PBE repair method that is guaranteed to generate only invulnerable regexes.  Importantly, our method can handle {\em real-world regexes} containing {\em lookarounds} and {\em backreferences}. 
Due to the extensions, the existing formal definitions of ReDoS vulnerabilities that only consider pure regexes are insufficient.  Therefore, we first give a novel {\em formal semantics and complexity of backtracking matching algorithms for real-world regexes}, and with them, give the {\em first formal definition of ReDoS vulnerability for real-world regexes}.
Next, we present a novel condition called {\em real-world strong 1-unambiguity} that is sufficient for guaranteeing the invulnerability of real-world regexes, and formalize the corresponding PBE repair problem.  
Finally, we present an algorithm that solves the repair problem. The algorithm builds on and extends the previous PBE methods to handle the real-world extensions and with constraints to enforce the real-world strong 1-unambiguity condition.
\end{abstract}

\begin{IEEEkeywords}
Real-world regexes, ReDoS, synthesis, repair
\end{IEEEkeywords}

\section{Introduction}
\label{sec:intro}
Regular expressions (regexes for short) have become an integral part of modern programming languages and software development, e.g., they are used as general purpose libraries~\cite{10.1145/2931037.2931073, RussCox}, for sanitizing user inputs~\cite{10.5555/2028067.2028068, 10.1145/2931037.2931050}, and extracting data from unstructured text~\cite{7374717, li-etal-2008-regular}.
Despite the widespread use of regexes in practice, it is an unfortunate fact that developers often 
write regexes which are vulnerable to {\em regex denial-of-service} (ReDoS) attacks in which attackers craft inputs that cause the regex matching algorithm to take super linear time~\cite{Davis:2018:IRE:3236024.3236027,ReDoS}.  ReDoS is a significant threat to our society due to the widespread use of regexes~\cite{Davis:2018:IRE:3236024.3236027,ReDoSIncident02, ReDoSIncident01, 217517}.
\tchanged{}{
While some regex engines offer mechanisms to limit their run time (directly by timeout or indirectly by limiting the number of backtrackings), determining a proper limit is often difficult, not to mention that vulnerable regexes may not even have any reasonable limits that can be assigned as they may struggle even on legitimate inputs.  Furthermore, such options are not available in many popular regex engines including those in the standard libraries of Python, Java, and Node.js.}

\tchanged{While}{To address the issue,} there has been much research on the topic of overcoming ReDoS vulnerability~\cite{10.1007/978-3-662-54580-5_1, Shen:2018:RCR:3238147.3238159,10.1007/978-3-319-40946-7_27, 10.1007/978-3-642-38631-2_11, SatoshiSugiyama2014,revealear}\tchanged{,}{.  However,} the previous works have focused mainly on the problem of detecting vulnerable regexes, and the problem of {\em repairing} them remains largely open.  As reported by Davis et al.~\cite{DavisReport, Davis:2018:IRE:3236024.3236027}, writing invulnerable regexes is a formidable task that developers often fail to achieve in practice.

Meanwhile, recent years have seen remarkable progress on {\em programming-by-example} (PBE) methods for synthesizing and repairing regexes~\cite{Alquezar94incrementalgrammatical,10.1145/3093335.2993244, 6994453, 7374717, 10.1145/3360565, DBLP:journals/corr/abs-1908-03316, FlashRegex}.
In these methods, a set of positive examples (strings to be accepted) and negative examples (strings to be rejected) are provided with the goal to synthesize a regex that correctly classify the examples, often with additional constraints to bias the synthesis toward ones syntactically close to the pre-repair regex~\cite{10.1145/3360565, FlashRegex}.
PBE methods have the salient advantage of easing the burden of writing correct regexes by letting the users reflect their intention by examples~\cite{10.1145/3093335.2993244, 6994453, 10.1145/3360565, FlashRegex}.
However, the existing PBE methods are not designed with resilience to ReDoS in mind and may generate vulnerable regexes.\footnote{\label{foot:flashregex} The only exception is the recent work by Li et al.~\cite{FlashRegex}, but they only handle pure regexes and also lack the guarantee to generate only invulnerable regexes (cf.~Section~\ref{sec:related}).}

In this paper, we rectify the situation by proposing the {\em first} PBE repair method that is {\em guaranteed to generate only invulnerable regexes}.
Importantly, our method can handle the so-called {\em real-world} regexes that have extensions such as {\em lookarounds}, {\em capturing groups}, and {\em backreferences}~\cite{masteringregex}.  

While previous works have investigated formal definitions of ReDoS vulnerability~\cite{10.1007/978-3-319-40946-7_27,SatoshiSugiyama2014,10.1007/978-3-662-54580-5_1}, they only address the pure regex fragment.  The overarching challenge in ReDoS vulnerabilities is to define the complexity of backtracking matching algorithm.  The previous works for pure regexes have used nondeterministic finite automata (NFA) to formalize the behavior of backtracking matching algorithms and its complexity.  Unfortunately, such a NFA-based definition is difficult for real-world regexes because the expressive power of real-world regexes is not regular~\cite{backreferenceisundeci}.  



Our first contribution is the first formal definition of ReDoS vulnerability for real-world regexes.  For this, we introduce a novel formal semantics of backtracking matching algorithm for real-world regexes and, by building on it, formally define the time complexity of backtracking matching algorithms for real-world regexes.  Also, we have discovered a subtle bug in a previous formal definition of ReDoS vulnerability for pure regexes~\cite{10.1007/978-3-662-54580-5_1} which can misclassify some vulnerable regexes as invulnerable (even for pure regexes).  Although the bug is fixable, this shows the subtlety of formalizing ReDoS vulnerability.

Our repair method ensures invulnerability by enforcing the novel {\em real-world strong 1-unambiguity} (\ltp{}) introduced in this paper.  \ltp{} is inspired by a notion for pure regexes called {\em strong 1-unambiguity}~\cite{10.1007/s00778-005-0169-1}, and can be considered as an extension of it to real-world regexes.  We show that \ltp{} is a sufficient condition for invulnerability, and formalize a PBE repair problem, {\em \ltp{} repair problem}, whose goal includes ensuring \ltp{}.  We prove that the \ltp{} repair problem is NP-hard.  We also show that a related notion for pure regexes called {\em 1-unambiguity} (also called {\em deterministic regex})~\cite{BRUGGEMANNKLEIN1998182,FlashRegex,10.1007/3-540-57273-2_45} is insufficient for guaranteeing invulnerability (even for pure regexes). 

Our third contribution is an algorithm for solving the \ltp{} repair problem.  Our algorithm builds on the previous PBE regex repair methods.  However, significant extensions are needed because {\em the previous methods neither support real-world regexes nor concern ReDoS vulnerability} (with the exception of \cite{FlashRegex} mentioned above).  A key step of the algorithm is generating SMT constraints that enforce both the \ltp{} condition and consistency with examples.  The latter is enforced by following our novel formal semantics of real-world regexes, and the former is enforced by using our novel {\em extended NFA translation} that is used to define \ltp{}.
We also adapt and extend the key techniques proposed for PBE regex synthesis and repair, such as the state space pruning technique by under- and over-approximations \cite{10.1145/3093335.2993244, 10.1145/3360565}, with the support for the real-world extensions and concerns for ReDoS vulnerability.

We have implemented a prototype of our algorithm in a tool called \tool{} (Regular Expression Modifier for Ensuring Deterministic propertY), and have experimented with the tool on a set of benchmarks of real-world regexes taken from~\cite{Davis:2018:IRE:3236024.3236027}.  The experimental results show that \tool{} was able to successfully repair non-trivial vulnerable regexes from a real-world data set. 

The contributions of the paper are summarized below.
\begin{itemize}
\item We initiate a study of ReDoS vulnerabilities for real-world regexes. To this end, we give a novel formal semantics and the time complexity of backtracking matching algorithms for real-world regexes, and with it, give the first formal definition of their ReDoS vulnerability.  We also show a subtle bug in a previous proposal for pure regexes~\cite{10.1007/978-3-662-54580-5_1}. (Section \ref{sec:regex})

\item We present the novel {\em real-world strong 1-unambiguity} (\ltp{}), and prove that the condition is sufficient for guaranteeing invulnerability for real-world regexes.  We define the \ltp{} repair problem and prove that the problem is NP-hard.  We also show that a related condition, 1-unambiguity (i.e., deterministic regex) for pure regexes~\cite{BRUGGEMANNKLEIN1998182,FlashRegex,10.1007/3-540-57273-2_45} is insufficient for ensuring invulnerability (even for pure regexes). (Section~\ref{sec:problem})

\item We give an algorithm for solving the \ltp{} repair problem that builds on and extends the previous PBE synthesis and repair methods.  Our algorithm extends the previous methods in two important ways: support for the real-world extensions and the incorporation of \ltp{} to enforce invulnerability. (Section~\ref{sec:algo})

\item We present an implementation of the algorithm in a tool called \tool{}, and present an evaluation of the tool on a set of real-world benchmarks. (Section~\ref{sec:eval})

\end{itemize}


\section{Overview}
\label{sec:example}

We give an informal overview of our repair algorithm by an example.  To illustrate, we use the regex  \texttt{<($\any^*$)$_1$>$\any^*$</$\backslash$1>} which is inspired by the one posted in~\cite{GTsBlog}.
The regex is intended to accept a {\em non-nested} XML tag, i.e., a tag that appears as a leaf in an XML document. 
%
For example, it should accept \texttt{<li></li>} and \texttt{<body>text</body>}, but it should reject \texttt{<li></body>}, \texttt{<body><li></li></body>}, and \texttt{<body><li></body>}.
Unfortunately, the regex is both incorrect and vulnerable.  It is incorrect because it accepts tags such as \texttt{<body><li></body>}.  It is vulnerable because it takes quadratic time to match strings such as
$\texttt{<}\texttt{><}\texttt{><}\:\cdots\:\texttt{></>}$ where $\cdots$ repeats \texttt{><}.  Indeed, running a regex engine such as Python's {\em re} on the regex will get stuck on suitably long strings of the above form.  



\tool{} can help the user automatically repair a regex like this into a correct invulnerable one.
To this end, the user provides the regex to be repaired along with sets of positive and negative examples.  Positive examples are strings that should be accepted, and negative examples are those that should be rejected.
%

{\flushleft\bf Sampling examples.}
As usual in a PBE scenario~\cite{10.1145/3360565, DBLP:journals/corr/abs-1908-03316, FlashRegex, 10.1145/3093335.2993244}, the user prepares test inputs that consists of positive and negative examples to validate the correctness of the regex.  Such examples may be prepared afresh by the user~\cite{8952499} or obtained from an existing collection such as RegExLib~\cite{regexlib}.
Generally, the result of PBE depends on the example selection.  Therefore, if the user cannot obtain an intended repair, she adds or removes examples and re-runs the tool to improve the result.
We note that, for usability, PBE should only use a relatively small number of examples.
For the running example, suppose that the user prepared positive examples \texttt{<ab></ab>} and \texttt{<a>ab</a>} and negative examples \texttt{<a></b>},  \texttt{<a><b></b></a>}, and \texttt{<a><ab></a>}.

\tool{} explores a regex that is consistent with the examples and has {\it real-world strong 1-unambiguity} (\ltp{}).  Also, \tool{} looks for regexes that are syntactically close to the given one to bias toward synthesizing regexes that are close to the user's intention.  The assumption is that the given regex may not be correct but is close to the one user intended.

%
\ltp{} ensures the invulnerability of the synthesized regex.  Roughly, it makes the behavior of the matching algorithm {\em backtrack-free} thus ensuring linear running time.
The regex \texttt{<($\any^*$)$_1$>$\any^*$</$\backslash$1>} violates the \ltp{} condition because there are two ways to match \texttt{>} in the input string after the first \texttt{<} is matched, that is, it can match the first $\any^*$ or the first \texttt{>}.
Likewise, after \texttt{</} is matched, there are again two ways to match \texttt{>}: $\backslash1$ if it refers to a string that starts with \texttt{>} or the second \texttt{>}.
There are also multiple ways to match \texttt{<} in the input string.
%
Next, we describe the steps of the repair process.

{\flushleft\bf Generating templates.}
\tool{} generates {\it templates}, which are regexes containing {\em holes}.  Informally, a hole $\hole{}$ is a placeholder that is to be replaced with some concrete regex.  \tool{} starts with the initial template set to be the input regex \texttt{<($\any^*$)$_1$>$\any^*$</$\backslash$1>}.  Since the regex is vulnerable and does not satisfy the \ltp{} condition, \tool{} replaces the subexpressions with holes and expands the holes by replacing them with templates such as $\hole\hole$, $\hole | \hole$, $\hole^*$, $(\text{?=}\hole)$, and $\backslash i$.
After some iterations, we get the template \texttt{<($\hole{}_1^*$)$_1$>$\hole{}_2^*$</$\backslash$1>}.


{\flushleft\bf Searching assignments.}
Next, \tool{} checks if the template can be instantiated to a regex that satisfies the required conditions by replacing its holes with some sets of characters.
For this, \tool{} generates two types of constraints: {\em consistency-with-examples constraint} that ensures that the regex is consistent with the examples, and {\em linear-time constraint} that asserts \ltp{}.  
\tool{} looks for a regex that satisfies the constraints by using an SMT solver.  If the constraints are unsatisfiable, then \tool{} backtracks to explore more templates.  
We give the details of the constraint generation in Section \ref{subsec:overview}. 
\tool{} also performs {\em template pruning} to filter out templates that can be efficiently detected impossible to be instantiated to a regex that is consistent with the examples. 
The details are presented in Section \ref{subsec:overview}.

Using an SMT solver, \tool{} finds that the constraints are satisfiable, and replaces $\hole_1$ and $\hole_2$ with $[\verb|^|\texttt{>}]$ and $[\verb|^|\texttt{<}]$, respectively.  Here, $[\verb|^|a]$ is a regex that matches any character besides $a$.  Finally, \tool{} returns  \texttt{<([}\textasciicircum{}\texttt{>]}$^*$\texttt{)}$_1$\texttt{>[}\textasciicircum{}\texttt{<]}$^*$\texttt{</}$\backslash$\texttt{1>} as the repaired regex which is invulnerable and matches the user's intention.

\section{Real-World Regular Expressions}
\label{sec:regex}
In this section, we give the definition of real-world regexes.  We also present the novel formal model of the backtracking matching algorithm for real-world regexes, and with it, we formally define their ReDoS vulnerability.


{\flushleft\bf Notations.}
Throughout this paper, we use the following notations.
We write $\Sigma$ for a finite alphabet; $a, b, c, \in \Sigma$ for a character; $w, x, y, z \in \Sigma^*$ for a sequence of characters; $\epsilon$ for the empty sequence; $r$ for a real-world regex; $\mathbb{N}$ for the set of natural numbers.
For the string $x = x[0]...x[n-1]$, its length is $|x| = n$.
For $0 \leq i \leq j < |x|$, the string $x[i]...x[j]$ is called a substring of $x$.
We write $x[i..j]$ for the substring.
In addition, we write $x[i..j)$ for the substring $x[i]...x[j-1]$.
We assume that $x[i..j) = \$$, where $\$ \notin \Sigma$, when $i < 0$ or $|x| < j$.
For $f$ a (partial) function,  $f[\alpha\mapsto\beta]$ denotes the (partial) function that maps $\alpha$ to $\beta$ and behaves as $f$ for all other arguments.  We write $f(\alpha) = \bot$ if $f$ is undefined at $\alpha$.
%
We define $\mathit{ite}(\mathit{true},A,B) = A$ and $\mathit{ite}(\mathit{false},A,B) = B$.


\subsection{Syntax and Informal Semantics}
\texcomment{
\begin{figure}[h]
\setlength\arraycolsep{2pt}
\[
\begin{array}{rcl}
r & ::= & [C] \mid \epsilon \mid rr \mid r|r \mid r^* \\
& \mid & (r)_i \mid \backslash i \mid \mbox{(?=$r$)} \mid \mbox{(?!$r$)} \mid \mbox{(?\textless=$x$)} \mid \mbox{(?\textless!$x$)} \\
\end{array}
\]
\caption{The syntax of real-world regular expressions}
\label{fig:syntax_of_regex}
\end{figure}
The syntax of {\em real-world regexes} (simply {\em regexes} or {\em expressions} henceforth) is given in Figure \ref{fig:syntax_of_regex}.
}

The syntax of {\em real-world regexes} (simply {\em regexes} or {\em expressions} henceforth) is given below:
\[
\begin{array}{rcl}
r & ::= & [C] \mid \epsilon \mid rr \mid r|r \mid r^* \\
& \mid & (r)_i \mid \backslash i \mid \mbox{(?=$r$)} \mid \mbox{(?!$r$)} \mid \mbox{(?\textless=$x$)} \mid \mbox{(?\textless!$x$)} \\
\end{array}
\]
Here, $C \subseteq \Sigma$ and $i \in \mathbb{N}$.
A set of characters $[C]$ exactly matches a character in $C$.
We sometimes write $a$ for $[\{a\}]$, and write $\any$ for $[\Sigma]$.  The semantics
of empty string $\epsilon$, concatenation $r_1 r_2$, union $r_1 | r_2$ and repetition $r^*$ are standard.  
Many convenient notations used in practice such as options, one-or-more repetitions, and interval quantifiers can be treated as syntactic sugars: $r? = r | \epsilon$, $r^+ = rr^*$, and $r\{i,j\} = r_1 ... r_i r_{i+1}? ... r_{j}?$ where $r_k = r$ for each $k \in \{1,\dots,j\}$.

The remaining constructs, that is, capturing groups, backreferences, (positive and negative) lookaheads and lookbehinds, comprise the real-world extensions.  In what follows, we will explain the semantics of the extended features informally in terms of the standard backtracking matching algorithm which attempts to match the given regex with the given (sub)string and backtracks when the attempt fails.  The formal definition is given later in the section.

A {\em capturing group} $(r)_i$ attempts to match $r$, and if successful, stores the matched substring in the storage identified by the index $i$.  Otherwise, the match fails and the algorithm backtracks.
A {\em backreference} $\backslash i$ refers to the substring matched to the corresponding capturing group $(r)_i$, and attempts to match the same substring if the capture had succeeded.  If the capture had not succeeded or the matching against the captured substring fails, then the algorithm backtracks.
For example, let us consider the regex \texttt{([0-9])$_1$([A-Z])$_2\backslash$1$\backslash$2}.
Here, \texttt{$\backslash$1} and \texttt{$\backslash$2} refer to the substring matched by \texttt{[0-9]} and \texttt{[A-Z]}, respectively.  The language represented by the regex is $\{ abab \mid a \in \texttt{[0-9]} \wedge b \in \texttt{[A-Z]} \}$.
Capturing groups in practice often do not have explicit indexes, but we write them here for clarity.
We assume without loss of generality that each capturing group always has a corresponding backreference and vice versa.
We assume that capturing group indexes are always distinct in a regex.

A {\em positive (resp. negative) lookahead} (?=$r$) (resp. (?!$r$)) attempts to match $r$ without any character consumption, and proceeds if the match succeeds (resp. fails) and backtracks otherwise.
A {\em fixed-string positive (resp. negative) lookbehind} (?\textless=$x$) (resp. (?\textless!$x$)) looks back (i.e., toward the left), attempts to match $x$ without any character consumption,
and proceeds if the match succeeds (resp. fails) or otherwise backtracks.
Fixed-string lookbehinds are supported by major regex engines such as those in Perl and Python~\cite{regexcookbook}.
Note that most regex engines do not support general lookbehinds~\cite{masteringregex}.

\subsection{Formal Semantics and Vulnerability}
\label{subsec:formalsemantics}

We now formally define the semantics of regexes.
Traditionally, the language of pure regexes is defined by
induction on the structure of the expressions.   However, such a definition would be difficult for real-world regexes because of the extended features and also unsuitable for formalizing vulnerability because the notion concerns the
complexity of backtracking matching algorithms.  To this end, we define the semantics by the matching relation ${\leadsto}$ that
models the behavior of backtracking matching algorithms.

A matching relation is of the form $(r, w, p, \Gamma) \leadsto{} \mathcal{N}$ where $p$ is a position on the string $w$ such that $0 \leq p \leq |w|$, $\Gamma$ is a function that maps each capturing group index to a string captured by the corresponding capturing group, and $\mathcal{N}$ is a set of matching results.  A {\em matching result} is a pair of a position and a capturing group function.
Roughly, $(r, w, p, \Gamma)$ is read:
a regex $r$ tries to match the string $w$ from the position $p$, with the information about capturing groups $\Gamma$.
For example, for the regex $a$ on the strings $a$ and $b$, the matching relations are $(a, a, 0, \emptyset{}) \leadsto{} \{ (1,\emptyset) \}$ and $(a, b, 0, \emptyset{}) \leadsto{} \emptyset{}$, respectively.
From these, the matching relation of the regex $(a|b)$ on the string $a$ is $((a|b), a, 0, \emptyset{}) \leadsto{} \{ (1,\emptyset) \}$.


\begin{figure}[t]\footnotesize
\infrule[Capturing group]
{(r, w, p, \Gamma) \leadsto{} \mathcal{N}}
{((r)_j, w, p, \Gamma) \leadsto{} \{ (p_i, \Gamma_i[ j \mapsto w[p..p_i) ]) \mid (p_i,\Gamma_i) \in \mathcal{N} \} }

\infrule[Backreference]
{\Gamma(i) \neq \bot \andalso (\Gamma(i),w,p,\Gamma) \leadsto{} \mathcal{N} }
{(\backslash i, w, p, \Gamma) \leadsto{} \mathcal{N} }

\infrule[Backreference Failure]
{ \Gamma(i) = \bot }
{(\backslash i, w, p, \Gamma) \leadsto{} \emptyset}


\infrule[Positive lookahead]
{(r, w, p, \Gamma) \leadsto{} \mathcal{N}  }
{(\text{(?=}r\text{)}, w, p, \Gamma) \leadsto{} \{ (p,\Gamma') \mid (\_,\Gamma') \in \mathcal{N} \} }

\infrule[Negative lookahead]
{(r, w, p, \Gamma) \leadsto{} \mathcal{N} \andalso{} \mathcal{N}' = \mathit{ite}(\mathcal{N} \neq \emptyset, \emptyset, \{(p,\Gamma)\})}
{(\text{(?!}r\text{)}, w, p, \Gamma) \leadsto{} \mathcal{N}' }

\infrule[Positive lookbehind]
{(x, w[p-|x|..p), 0, \Gamma) \leadsto{} \mathcal{N} \andalso \mathcal{N}' = \mathit{ite}(\mathcal{N} \neq \emptyset, \{(p,\Gamma)\}, \emptyset)}
{(\text{(?\textless=}x\text{)}, w, p, \Gamma) \leadsto{} \mathcal{N}' }

\infrule[Negative lookbehind]
{(x, w[p-|x|..p), 0, \Gamma) \leadsto{} \mathcal{N} \andalso \mathcal{N}' = \mathit{ite}(\mathcal{N} \neq \emptyset, \emptyset, \{(p,\Gamma)\})}
{(\text{(?\textless!}x\text{)}, w, p, \Gamma) \leadsto{} \mathcal{N}' }

\caption{Selected rules of the matching relation $\leadsto$}
\label{fig:semanticsv} 
\end{figure}

Figure~\ref{fig:semanticsv} shows some rules for deducing the matching relation.  For space, we show only the rules for handling
the extended features, deferring the full rules to the \supplementary{}.  The rules are inspired by \cite{DBLP:journals/scp/MedeirosMI14} who have given natural-semantics-style rules for pure regexes and parsing expression grammars.  However, to our knowledge, we are the first to give the formal semantics of real-world regexes in this style and use it to formalize vulnerability.

%
In the rule (\textsc{Capturing group}), we first get the matching result $\mathcal{N}$ from matching $w$ against $r$ at the current position $p$.  And for each matching result $(p_i,\Gamma_i) \in \mathcal{N}$ (if any), we record the matched substring $w[p..p_i)$ in the corresponding capturing group map $\Gamma_i$ at the index $i$.  The rule (\textsc{Backreference}) looks up the captured substring and tries to match it with the input at the current position.  The match fails if the corresponding capture has failed as stipulated by the rule (\textsc{Backreference Failure}).

In the rule (\textsc{Positive lookahead}), the expression $r$ is matched against the given string $w$ at the current position $p$ to obtain the matching results $\mathcal{N}$.  Then, for every match result $(p',\Gamma') \in \mathcal{N}$ (if any), we reset the position from $p'$ to $p$.  This models the behavior of lookaheads which does not consume the string.  The rule (\textsc{Negative lookahead}) is similar, except that we reset and proceed when there is no match.  Note that captures made inside of a negative lookahead cannot be referred outside of the lookahead, which agrees with the behavior of regex engines in practice.  The rules (\textsc{Positive lookbehind}) and (\textsc{Negative lookbehind}) for handling fixed-string lookbehinds are self-explanatory.

\begin{definition}[Language]
\normalfont
The {\em language} of a regex $r$ is defined as
$L(r) = \{ w \mid (r, w, 0, \emptyset) \leadsto \mathcal{N} \wedge \exists\Gamma.(|w|,\Gamma) \in \mathcal{N} \}$.
\end{definition}

We show some examples of matchings.  For brevity, we omit capturing group information from Examples \ref{ex:pr1} and \ref{ex:pr2} because it is not used there, i.e., it is always $\emptyset$.
\begin{example}
\label{ex:pr1}
\normalfont
The matching of the regex $(a^*)^*$ on the string $ab$ is as follows:
\vspace{-1em}
\begin{prooftree}
\small
\insertBetweenHyps{\hspace{-2pt}}
\AxiomC{$0 < |ab|$}
\AxiomC{$a \in \{a\}$}
\BinaryInfC{$(a, ab, 0) \leadsto{} \{1\}$}
\insertBetweenHyps{\hspace{-2pt}}
\AxiomC{$1 < |ab|$}
\AxiomC{$b \notin \{a\}$}
\BinaryInfC{$(a, ab, 1) \leadsto{} \emptyset$}
\UnaryInfC{$(a^*,ab,1) \leadsto{} \{1\}$}
\BinaryInfC{$(a^*,ab,0) \leadsto{} \{0,1\}$}
\insertBetweenHyps{\hspace{-2pt}}
\AxiomC{$1 < |ab|$}
\AxiomC{$b \notin \{a\}$}
\BinaryInfC{$(a,ab,1) \leadsto{} \emptyset$}
\UnaryInfC{$((a^*)^*, ab, 1) \leadsto{} \{1\}$}
\BinaryInfC{$((a^*)^*, ab, 0) \leadsto{} \{0, 1\}$}
\end{prooftree}
The regex rejects the string because $|ab| = 2 \notin \{0,1\}$.
\end{example}
\begin{example}
\label{ex:pr2}
\normalfont
The matching of $((\mbox{?=}a)^*)^*$ on $ab$ is:
\begin{prooftree}
\small
\insertBetweenHyps{\hspace{-2pt}}
\AxiomC{$0 < |ab|$}
\AxiomC{$a \in \{a\}$}
\BinaryInfC{$(a, ab, 0) \leadsto{} \{1\}$}
\UnaryInfC{$((\mbox{?=}a), ab, 0) \leadsto{} \{0\}$}
\UnaryInfC{$((\mbox{?=}a)^*, ab, 0) \leadsto{} \{0\}$}
\UnaryInfC{$((((\mbox{?=}a)^*)^*), ab, 0) \leadsto{} \{0\}$}
\end{prooftree}
The regex rejects the string because $|ab| = 2 \notin \{0\}$.
\end{example}

\begin{example}
\label{ex:pr3}
\normalfont
The matching of $(a^*)_1\backslash1$ on $aa$ is:
\[
\small
\infer{(a^*)_1\backslash1 \leadsto{} \{(0,\Gamma_0), (2,\Gamma_1)\} }
    { A & B_0 & B_1 & B_2}
\]
where $\Gamma_0 = \{(1,\epsilon)\}$, $\Gamma_1 = \{(1,a)\}$, $\Gamma_2 = \{(1,aa)\}$ and the subderivation
$A$ is:
\begin{prooftree}
\small
\insertBetweenHyps{\hspace{-2pt}}
\AxiomC{$0 < |aa|$}
\AxiomC{$a \in \{a\}$}
\BinaryInfC{$(a,aa,0,\emptyset)\leadsto{}\{(1,\emptyset)\}$}
\insertBetweenHyps{\hspace{-2pt}}
\insertBetweenHyps{\hspace{-2pt}}
\AxiomC{$C_0$}
\AxiomC{$C_1$}
\UnaryInfC{$(a^*,aa,2,\emptyset{}) \leadsto{} \emptyset{}$}
\BinaryInfC{$(a^*,aa,1,\emptyset) \leadsto{} \{(1,\emptyset),(2,\emptyset)\}$}
\BinaryInfC{$(a^*,aa,0,\emptyset) \leadsto{} \{(0,\emptyset), (1,\emptyset), (2,\emptyset)\}$}
\UnaryInfC{$((a^*)_1,aa,0,\emptyset{}) \leadsto{} \{(0,\Gamma_0),(1,\Gamma_1),(2,\Gamma_2)\}$}
\end{prooftree}
and the roots of the subderivations $B_0$, $B_1$, $B_2$, $C_0$, $C_1$ are, respectively, $(\backslash 1,aa,0,\Gamma_0) \leadsto \{(0,\Gamma_0)\}$, $(\backslash1,aa,1,\Gamma_1) \leadsto \{(2,\Gamma_1)\}$, $(\backslash1,aa,2,\Gamma_2) \leadsto \emptyset$,
$(a,aa,1,\emptyset{}) \leadsto{} \{(2,\emptyset{})\}$, $(a,aa,2,\emptyset{}) \leadsto{} \emptyset{}$.
The regex accepts the string as $(|aa|,\Gamma_1) \in \{(0,\Gamma_0), (2,\Gamma_1)\}$.




%
\end{example}


We define the {\em size} of the derivation a matching relation to be the number of nodes in the derivation tree.  Note that the size is well defined because our rules are deterministic.
\begin{definition}[Running time]
\normalfont
For a regex $r$ and a string $w$, we define the {\em running time} of the backtracking matching algorithm on $r$ and $w$, $\runtime(r,w)$, to be the size of the derivation of $(r, w, 0, \emptyset) \leadsto \mathcal{N}$.
\end{definition}
%
\begin{definition}[Vulnerable Regular Expressions]
\label{def:vulnerable}
\normalfont
We say that an expression $r$ is {\em vulnerable} if $\runtime(r,w) \notin O(|w|)$.
\end{definition}
Note that a regex $r$ is vulnerable iff there exist infinitely many strings $w_0$, $w_1$,\dots such that $\runtime(r,w_i)$ (for $i \in \mathbb{N}$) grows super-linearly in $|w_i|$. Such strings are often called {\em attack strings}.
For example, $(a^*)^*$ in Example~\ref{ex:pr1} 
and $(a^*)_1\backslash1$ in Example~\ref{ex:pr3} are vulnerable because there exist attack strings $\{ a^nb \mid n \in \mathbb{N} \}$ on which $(a^*)^*$ and $(a^*)_1\backslash1$ respectively take $\Omega(n!)$ and $\Omega(n^2)$ time.  Indeed, running an actual regex engine such as Python's \textit{re} on these regexes with these attack strings exhibits a super-linear behavior.  By contrast, $((\mbox{?=}a)^*)^*$ in Example~\ref{ex:pr2} takes $O(n)$ time on these strings and is in fact invulnerable.  

Our matching semantics captures the behavior of common backtracking matching algorithms used in most real regex engines, e.g., ones based on path traversal of some non-deterministic automaton~\cite{10.1007/978-3-662-54580-5_1, Shen:2018:RCR:3238147.3238159, revealear}. 
We remark that our formal semantics may be less efficient than an actual regex engine because it computes all possible runs without any optimization.  However, it is sound for defining invulnerability, and our repair algorithm synthesizes regexes that are invulnerable even with respect to the inefficient formal semantics.
This implies that if a pure regex is considered vulnerable according to the definition of vulnerability in \cite{10.1007/978-3-662-54580-5_1} then it is also considered vulnerable according to our definition.

It is worth noting that $(a^*)^*$ is incorrectly classified as invulnerable by~\cite{10.1007/978-3-662-54580-5_1}, both according to their formal definition of vulnerability and by their vulnerability detection tool.  Although the bug is fixable by adding $\epsilon$ transitions to their NFA-based definition in a certain way, this shows the subtlety of formalizing vulnerability.

\section{\ltp{} and Its Repair Problem}
\label{sec:problem}

This section presents our PBE repair algorithm.  First, we define the novel notion of {\em real-world strong 1-unambiguity} (\ltp{}) and prove it to be sound for ensuring invulnerability (Section~\ref{subsec:lineartime}).  Then, we define {\em \ltp{} repair problem} to be the problem of synthesizing a regex that correctly classifies the given positive and negative examples, satisfies \ltp{}, and is syntactically close to the pre-repair regex (Section~\ref{subsec:repairproblem}).  We prove that the \ltp{} repair problem is NP-hard.  Section~\ref{sec:algo} presents an algorithm for solving the \ltp{} repair problem.

\subsection{Real-World Strong 1-Unambiguity}
\label{subsec:lineartime}
We begin by introducing some preliminary notions.
\begin{definition}[Bracketing]
\normalfont
The {\it bracketing} of $r$, $r^{[]}$, is obtained by inductively mapping each subexpression $s$ of $r$ to $[_i s ]_i$ where $i$ is a unique index.  Here, $[_i$ and $]_i$ are called {\em brackets} and are disjoint from the alphabet $\Sigma$ of $r$.
\end{definition}
Note that $r^{[]}$ is a regex over the alphabet $\Sigma \cup \Bracketing$, where $\Bracketing = \{ [_i, ]_i | i \in \mathbb{N}\}$.  We call $\Bracketing$ the {\em bracketing alphabet} of $r^{[]}$.
For example, for $r = ((a)^*)^*b$, the bracketing is 
\[
r^{[]} = [_1 [_2 ( [_3 ( [_4 a ]_4 )^* ]_3 )^* ]_2 [_5 b ]_5 ]_1
\]
with the bracketing alphabet $\{ [_i, ]_i \mid i \in \{ 1,2,3,4,5 \} \}$. 
\begin{definition}[Lookaround removal]
\label{def:removela}
\normalfont
The regex $r$ with its {\em lookarounds removed}, $\removela(r)$, is $r$ but with each of its lookaround replaced by $\epsilon$. 
\end{definition}

\begin{figure*}
\begin{center}
\small
\begin{eqnarray*}
\toENFA([C]) & = & (\{q_0,q_1\}, \{(q_0, a, q_1) \mid \forall a \in C \}, q_0, q_1 )\\
\toENFA(r_1r_2) & = & (Q_1 \cup Q_2, \delta_1 \cup \delta_2 \cup \{ (q_{n_1}, \epsilon, q_{0_2}) \}, q_{0_1}, q_{n_2} )
	\ \ \textit{where} \ \ (Q_1, \delta_1, q_{0_1}, q_{n_1}) = \toENFA(r_1) \textit{ and } (Q_2, \delta_2, q_{0_2}, q_{n_2}) = \toENFA(r_2)\\	
\toENFA(r_1 | r_2) & = & (Q_1 \cup Q_2 \cup \{ q_0, q_n \}, \delta_1 \cup \delta_2 \cup \{ (q_0, \epsilon, q_{0_1}), (q_0, \epsilon, q_{0_2}), (q_{n_1}, \epsilon, q_n), (q_{n_2}, \epsilon, q_n) \}, q_0, q_n)  \\
	&& \textit{where} \ \ (Q_1, \delta_1, q_{0_1}, q_{n_1}) = \toENFA(r_1) \textit{ and } (Q_2, \delta_2, q_{0_2}, q_{n_2}) = \toENFA(r_2)\\	
\toENFA(r^*) & = & (Q \cup \{ q_0, q_n \}, \delta \cup \{ (q_0, \epsilon,q_{0_1}), (q_0, \epsilon, q_n), (q_{n_1}, \epsilon, q_n), (q_{n_1}, \epsilon, q_{0_1}) \}, q_0, q_n) \ \ \textit{where} \ \ (Q, \delta, q_{0_1}, q_{n_1}) = \toENFA(r)\\
\toENFA((r)_i) & = & \toENFA(r) \textit{ and } \mathcal{I} = \mathcal{I}[i \mapsto q_0] \ \ \textit{where} \ \ \toENFA(r) = (\_,\_,q_0,\_) \\
\toENFA(\backslash i) & = & (\{q_0, q_1\}, \{ (q_0, a, q_1) ~|~ a \in \text{\follow($\mathcal{I}(i)$)$^{\natural}$} \} \cup \{ (q_0, \epsilon, q_1) ~|~ (r)_i \mbox{ and } \epsilon \in L(r) \}, q_0, q_1 )
\end{eqnarray*}
\end{center}
\caption{The extended NFA translation.}
\label{fig:toenfa}
\end{figure*}

A {\em non-deterministic automaton} (NFA) over an alphabet $\Sigma$ is a tuple $(Q, \delta, q_0, q_n)$ where $Q$ is a finite set of states, $\delta \subseteq Q \times (\Sigma\cup\{ \epsilon \}) \times Q$ is the transition relation, $q_0$ is the initial state, and $q_n$ is the accepting state. 
\begin{definition}[$\toENFA$]
\label{def:toenfa}
\normalfont
For a lookaround-free regex $r$ over $\Sigma$, its {\em extended NFA translation}, $\toENFA(r^{[]})$, is a NFA over $\Sigma \cup \Bracketing$ defined by the rules shown in Figure~\ref{fig:toenfa} where $\Bracketing$ is the bracketing alphabet of $r^{[]}$.
\end{definition}
In the translation shown in Figure~\ref{fig:toenfa}, we maintain a global map $\mathcal{I}$ from capturing group indexes to states.  $\mathcal{I}$ is initially empty and is updated whenever a capturing group $(r)_i$ is encountered so that $\mathcal{I}(i)$ is set to be the initial state of the NFA constructed from $r$.  
$\follow(q)$ is defined as follows: $\rho a \in \follow(q)$ iff $\rho \in \Bracketing^*$, $a \in \Sigma$, and there is a $\rho a$-labeled path from $q$.  We define $\rho a^{\natural} = a$, and \follow$(q)^{\natural}$ = $\{ a \mid \rho a \in \text{\follow}(q) \}$.
Roughly, \follow{}$(q)^{\natural}$ is the set of characters that $r$ can reach without any character consumption where $q$ is the initial state of $\toENFA{}(r)$. For example, for $r = ab|ac|d^*ef$, \follow$(q)^{\natural} = \{ a,d,e \}$ where $q$ is the initial state of $\toENFA{}(r)$.

Our extended NFA translation may be seen as the standard Thompson's translation for pure regexes~\cite{10.1145/363347.363387,DBLP:books/daglib/0086373} extended to real-world regexes.  However, unlike the Thompson's translation, it does not preserve the semantics (necessarily not so because real-world regexes are not regular even without lookarounds).  Instead, we use the translation only for the purpose of defining \ltp{}.  
%
For a pair of states $q$ and $q'$ of a NFA, we write
$\textit{paths}(q,q')$ for the set of strings that take the NFA from $q$ to $q'$.
\begin{definition}[$\computePath$]
\normalfont
For $r$ a regex over $\Sigma$,  $\Bracketing$ the bracketing alphabet of $r^{[]}$, $[_i\: \in \Bracketing$, $a \in \Sigma$, and $(\_,\delta,\_,\_) = \toENFA(r^{[]})$, we define $\computePath(r,[_i, a)$ to be the set below:
\[
\{ \rho \in \Bracketing^* \mid \exists (q_j, [_i, \_), (q_l, a, \_) \in \delta. \rho \in paths(q_j,q_l) \}.
\]

\end{definition}
Roughly, $\computePath(r,[_i, a)$ are the sequences of brackets appearing in paths from
the unique edge labeled $[_i$ to an edge labeled $a$ in the extended NFA translation of $r$.

\begin{example}
\label{ex:astarstarenfa}
\normalfont
\begin{figure}[t]
  \centering
    \includegraphics[clip,width=7.0cm]{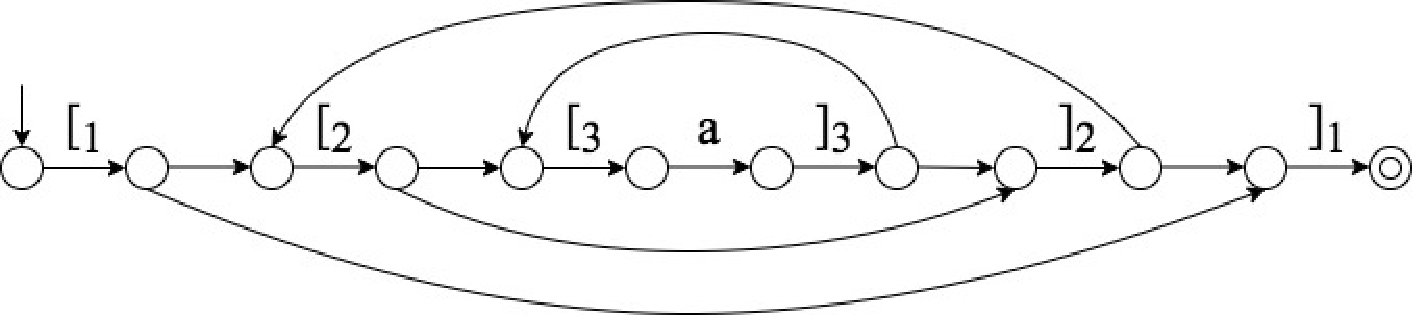}
     \caption{The extended NFA translation of $(a^*)^*$.}
    \label{fig:enfaexample} 
\end{figure} 
Figure~\ref{fig:enfaexample} shows the extended NFA translation of $(a^*)^*$ where unlabeled edges denote $\epsilon$ transitions.  Note that $\computePath((a^*)^*, [_1, a) = \{ [_1([_2 ]_2)^n [_2 [_3 \;\mid n \in \mathbb{N} \}$.
\end{example}

\begin{definition}[\ltp{}]
\label{def:ltp}
\normalfont
We say that a regex $r$ satisfies {\em real-world strong 1-unambiguity} (\ltp{}) if (1) $|\computePath(\removela(r), [_i, a)| \leq 1$ for all $a \in \Sigma$ and $[_i\: \in \Bracketing$ where $\Bracketing$ is the bracketing alphabet of $\removela(r)^{[]}$ and (2) lookarounds in $r$ do not contain repetitions and backreferences.
\end{definition}
Roughly, condition (1) ensures that the matching algorithm can determine which subexpression to match next by looking at the next character in the input string.  It therefore rules out the need for backtracking.  The condition is inspired by a notion called {\em strong 1-unambiguity} for pure regexes~\cite{10.1007/s00778-005-0169-1} and can be seen as an extension of it to regexes containing backreferences.
We do not impose the condition in lookarounds, because the condition prohibits some important use patterns of them.  For instance, it will preclude any meaningful use of a positive lookahead because if the lookahead succeeds then the subexpression immediately following the lookahead must match the same string.  Therefore, for lookarounds, we impose the condition stipulated by (2).  The condition prohibits repetitions and backreferences to appear in a lookaround and ensures that the matching within a lookaround finishes in constant time.  Therefore, (1) and (2) combined guarantee that the overall matching finishes in linear time.
\begin{example}
\normalfont
Recall $r_1 = (a^*)^*$, $r_2 = ((\mbox{?=}a)^*)^*$, $r_3 = (a^*)_1 \backslash1$ from Examples~\ref{ex:pr1}, \ref{ex:pr2}, \ref{ex:pr3}.
The regex $r_1$ does not satisfy the \ltp{} condition because as shown in Example~\ref{ex:astarstarenfa}, $|\computePath(r_1, [_1, a)| = \aleph_0 > 1$.  Also, $r_3$ does not satisfy the \ltp{} condition because
$\computePath(r_3, [_1, a) = \{ [_1[_2[_3[_4, [_1[_2[_3]_3]_2[_5 \}$ where 
\[
r_3^{[]} = [_1 [_2 ( [_3 ([_4 a ]_4)^* ]_3 )_1 ]_2\: [_5 \backslash 1 ]_5 ]_1,
\]
and so $|\computePath(r_3, [_1, a)| = 2 > 1$.  By contrast, $r_2$ (trivially) satisfies the \ltp{} condition because $\removela(r_2) = (\epsilon^*)^*$ which contains no characters. 
\end{example}

\begin{example}
\normalfont
The regex $r_4 = a^*b^*$ satisfies the \ltp{} condition because $|\computePath(r_4, [_i, a)| = |\computePath(r_4, [_i, b)| = 1$ for $i \in \{1,2,3\}$, and $|\computePath(r_4, [_i, a)| = 0$ and $|\computePath(r_4, [_i, b)| = 1$ for $i \in \{4,5\}$, where 
$r_4^{[]} = [_1[_2([_3 a ]_3)^*]_2\: [_4([_5 b ]_4)^*]_5 ]_1$.  The regex $r_5 = ((\mbox{?=}\any^*)\any)^*$  does not satisfy the \ltp{} condition because the positive lookahead contains a repetition, violating condition (2).
\end{example}

We show that \ltp{} is a sufficient condition for invulnerability.
\begin{theorem}
\label{theo:numstlin}
A regex that satisfies \ltp{} is invulnerable.
\end{theorem}
The proof appears in the \supplementary{}.
%
We remark that while \ltp{} is a sufficient condition, it is not a necessary condition for invulnerability.  For example, $a|aa$ is invulnerable but does not satisfy \ltp{}.  

Finally, we note that a related notion called {\em 1-unambiguity} for pure regexes (also called {\em deterministic regexes})~\cite{BRUGGEMANNKLEIN1998182,FlashRegex,10.1007/3-540-57273-2_45} is insufficient for guaranteeing invulnerability (even for pure regexes).  For example, $(a^*)^*$ is 1-unambiguous, because any character occurs at most once, but it is vulnerable as shown in Section~\ref{subsec:formalsemantics}\footnote{Further details are in Appendix~\ref{appendix:insufficiency_of_flashregex}.}.


\subsection{Repair Problem}
\label{subsec:repairproblem}
In this section, we define the \ltp{} repair problem.
First, we adapt the notion of {\em distance} between regexes from a recent work on
PBE regex repair~\cite{10.1145/3360565}.  In what follows, a regex is identified with its {\em abstract syntax tree} (AST) representation.  For an AST $r$, we define its {\em size}, $|r|$, to be the number of nodes of $r$.
\begin{definition}[Distance]
\label{def:dist}
\normalfont
For non-overlapping subtrees $r_1$, \dots, $r_n$ of a regex $r$, an {\it edit} $r[r_1'/r_1, \cdots, r_n'/r_n]$ replaces each $r_i$ with $r_i'$.
The {\em cost} of the edit is $\sum_{i \in \{1,\dots,n\}} |r_i| + |r_i'|$. The {\em distance} between  $r_1$ and $r_2$, $\distfunc{}(r_1,r_2)$, is the minimum cost of an edit that transforms $r_1$ to $r_2$.
\end{definition}
For example, $\distfunc{}(a|b|c,d|c) = 4$, which is realized by the edit that replaces $a|b$ by $d$.
We now define the repair problem.
\begin{definition}[\ltp{} Repair Problem]
\label{def:ltprepair}
\normalfont
Given a regex $r_1$, a finite set of {\em positive examples} $P \subseteq \Sigma^{*}$, and a finite set of {\em negative examples} $N \subseteq \Sigma^{*}$ where $P \cap N = \emptyset$, the {\em real-world strong 1-unambiguity repair problem} ({\em \ltp{} repair problem}) is the problem of synthesizing $r_2$ such that (1) $r_2$ satisfies \ltp{}, (2) $P \subseteq L(r_2)$, (3) $N \cap L(r_2) = \emptyset$, and (4) $\distfunc{}(r_1,r_2) \leq \distfunc{}(r_1,r_3)$ for any regex $r_3$ satisfying (1)-(3).
\end{definition}
Condition (1) guarantees that the repaired regex $r_2$ is invulnerable.  Conditions (2) and (3) assert that $r_2$ correctly classifies the examples.  Condition (4) says that $r_2$ is syntactically close to the original regex $r_1$.   

We note that the repair problem is easy without the closeness condition (4): one can construct an invulnerable regex that accepts just $P$ (or $\Sigma^*\setminus N$) in time linear in $\sum_{w \in P} |w|$ (or $\sum_{w \in N} |w|$).  However, such a regex is unlikely to be one intended by the user, that is, it suffers from {\em overfitting}.  Condition (4) is an important ingredient of a PBE synthesis and repair that biases the solution toward the intended one.  The assumption is that the given regex may not be quite correct but is close to the one user intended.

We show that the \ltp{} repair problem is NP-hard by a reduction from \textsc{ExactCover} which is NP-complete \cite{Karp1972}.  More formally, we consider the decision problem version of the \ltp{} repair problem in which we are asked if there is a repair $r_2$ of $r_1$ satisfying conditions (1)-(3) and $\distfunc{}(r_1,r_2) \leq k$ for some given $k \in \mathbb{N}$.  Note that the decision problem is no harder than the original repair problem because the solution to the repair problem can be used to solve the decision problem.
\begin{theorem}
\label{theo:nphardness}
The \ltp{} repair problem is NP-hard.
\end{theorem}
The proof appears in the \supplementary{}.



\section{Repair Algorithm}
\label{sec:algo}

\begin{algorithm}[t]
\caption{The repair algorithm}
\label{alg:repair}
\begin{algorithmic}[1]
\REQUIRE regex $r$, positive examples $P$, negative examples $N$
\ENSURE a \ltp{} regex that is consistent with $P$ and $N$
\STATE \queue{} $ \leftarrow $ \{ $r$ \}
{\STATE {\bf while}  \queue{} is not empty {\bf do}}
\STATE \idt{} \sstate{} $\leftarrow$ \queue{}{\sf .pop()}
\STATE \idt{} {\bf if} $P \subseteq L(\sstate{}_{\top})$ {\bf and} $N \cap L(\sstate{}_{\bot}) = \emptyset$ {\bf then}
\STATE \idt{}\idt{} $\Phi$ $\leftarrow$ {\sf getInvulnerableConstraint}(\sstate{}, $P$, $N$)
\STATE \idt{}\idt{} {\bf if} $\Phi$ is satisfiable {\bf then}
\STATE \idt{}\idt{}\idt{} {\bf return} {\sf solution(\sstate{}}, $\Phi${\sf)}
\STATE \idt{}\idt{} \queue{}{\sf.push(expandHoles(\sstate{}))}
\STATE \idt{} \queue{\sf .push(addHoles(\sstate{}))}
\end{algorithmic}
\end{algorithm}
In this section, we describe the details of our PBE repair algorithm.
As discussed in Section~\ref{sec:example}, our algorithm builds on the previous approaches that use template-based search with search pruning~\cite{10.1145/3093335.2993244, 10.1145/3360565} and the SMT-based constraint solving to find a solution within the given candidate template~\cite{10.1145/3360565}.  
Our algorithm extends the constraint generations and the pruning techniques of the previous approaches with the support for real-world extensions and the assertion of \ltp{} to ensure invulnerability.
We give the overview of the repair algorithm in Section~\ref{subsec:overview}.  The details of the constraint generation is given in Section~\ref{subsec:genconst}.


\subsection{Algorithm Overview}
\label{subsec:overview}
Algorithm \ref{alg:repair} shows the high-level structure of the repair algorithm.
The algorithm takes a regex $r$, a set of positive examples $P$, and a set of negative examples $N$ as input.
Its output is a regex that satisfies the \ltp{} condition and is consistent with $P$ and $N$.
At a high level, our algorithm consists of the following four key components.

{\flushleft\bf Generate the initial template.}
The priority queue \queue{} maintains regex templates.  A {\em regex template} \sstate{} is a regex that may contain a {\em hole} $\hole{}$ denoting a placeholder that is to be replaced by a concrete regex.  Its syntax is formally the extension of that of regexes (cf.~Section~\ref{sec:regex}) and is defined by: $r ::= \cdots \mid \hole{}$.  To distinguish, we will use $\sstate{}$ to range over regex templates and reserve $r$ for concrete regexes.


The queue \queue{} is initialized by pushing the input regex (line 1).
\queue{} ranks its elements by the distance defined in Section \ref{sec:problem} so that templates closer to the input regex are placed before. 
Due to this, \tool{} outputs a regex that satisfies condition (4) of \ltp{} repair problem, i.e., the regex is minimal. 


{\flushleft\bf Pruning by approximations.}
The algorithm next retrieves and removes a template $\sstate$ from the head of \queue{} (line 3), and applies the {\it feasibility check} to the template (line 4).
The feasibility check is introduced by~\cite{10.1145/3093335.2993244} for pure regexes.  It is known to substantially reduce the search space and is also used in subsequent works on PBE regex synthesis and repair~\cite{10.1145/3360565, DBLP:journals/corr/abs-1908-03316}.  We extend the idea with the support for the real-world features.

The over- and under-approximation $\sstate{}_\top$ and $\sstate{}_\bot$ are built to satisfy the properties $L(r') \subseteq L(\sstate{}_\top)$ and $L(\sstate{}_\bot) \subseteq L(r')$ for any regex $r'$ obtainable by filling the holes of $\sstate$.  If $P \nsubseteq L(\sstate{}_\top)$ or $N \cap L(\sstate{}_\bot) \not= \emptyset$, then there is no way to get a regex consistent with $P$ and $N$ from the template, and thus we safely discard the template from the search.

The approximations are built by filling each hole in $\sstate{}$ with either 
$\any^*$ or $[\emptyset]$ based on whether an under- or over- approximation is to be made and whether the hole appears in even or odd number of negative lookarounds.  Let $\overline{\any^*} = [\emptyset]$ and $\overline{[\emptyset]} = \any^*$.  Then, $\sstate{}_\top = \tapprox(\sstate{},\any^*)$ and $\sstate{}_\bot = \tapprox(\sstate{},[\emptyset])$ where $\tapprox(t,r)$ is inductively defined as follows:
\[
\setlength\arraycolsep{2pt}
\begin{array}{rclcrcl}
\tapprox(\hole,r) & = & r & &\tapprox([C],r) & = & [C] \\
\tapprox(\sstate{}_1\sstate{}_2,r) & = & \tapprox(\sstate{}_1,r)\tapprox(\sstate{}_2,r) & & \tapprox(\epsilon,r) & = & \epsilon \\
\tapprox(\sstate{}_1|\sstate{}_2,r) & = & \tapprox(\sstate{}_1,r)|\tapprox(\sstate{}_2,r) & & \tapprox(\sstate{}^*,r) & = & \tapprox(\sstate{},r)^* \\
\tapprox((\sstate{})_i,r) & = & (\tapprox(\sstate{},r))_i & & \tapprox(\backslash i,r) & = & \backslash i \\
\tapprox(\mbox{(?=$\sstate{}$)},r) & = & \mbox{(?=$\tapprox(\sstate{},r)$)} & &\tapprox(\mbox{(?!$\sstate{}$)},r) & = & \mbox{(?!$\tapprox(\sstate{},\overline{r})$)} \\
\tapprox(\mbox{(?\textless=$x$)},r) & = & \mbox{(?\textless=$x$)} & & \tapprox(\mbox{(?\textless!$x$)},r) & = & \mbox{(?\textless!$x$)} 
\end{array}
\]

{\flushleft\bf Searching assignments by constraints solving.}
If the feasibility check passes, the algorithm decides if the template can be instantiated into a regex that is consistent with the examples and satisfies the \ltp{} condition by filling each hole with a set of characters (i.e., some $[C]$).  This is done by encoding the search problem as a constraint satisfaction problem which is then solved by an SMT solver (lines 5-6).
%
We defer the details of this phase to Section~\ref{subsec:genconst}.

{\flushleft\bf Expanding and adding holes to a template.}
The failure of the SMT solver to find a solution implies that there exists no instantiation of the
template obtainable by filling the holes by sets of characters that is consistent with the examples and satisfies the \ltp{} condition.  In such a case, our algorithm expands the holes in the template to generate unexplored templates and add them to the queue (line 8).  For example, the template $(\hole{})_1\backslash 1$ is expanded to $(\hole{}\hole{})_1 \backslash 1$, $(\hole{}|\hole{})_1 \backslash 1$, $(\hole{}^*)_1 \backslash 1$, and so on. 
Here, to ensure the \ltp{} condition, we do not replace the holes in lookarounds with templates containing repetitions.  

Finally, if the current template fails the feasibility check and no more templates are in the queue, we generate new templates by adding holes to the current template and add the new templates to the queue, because it would be fruitless to expand the current template any further (line 9).  The addition of a new hole is done by replacing a set of characters by a hole or replacing an expression with a hole when an immediate subexpression of the expression is a hole.
Note that changing an operator is possible because {\sf addHoles} can replace an operator with a hole when an immediate subexpression is a hole, and then {\sf expandHoles} can replace the hole with a different operator. For example, by this, $(a|b)c$ may be repaired to $d^*c$.


\subsection{Generating Constraints}
\label{subsec:genconst}
We show the construction of the SMT constraint.  The constraint is a conjunction of the following two constraints: the {\em consistency-with-examples constraint} which asserts that regex obtained by replacing the holes in the template with the sets of characters is consistent with the given positive and negative examples, and the {\em linear-time constraint} which further constrains such a regex to satisfy the \ltp{} condition. 
We describe the constructions of each constraint in Section \ref{subsec:consthole} and \ref{subsec:oneprop}, respectively.
%



\subsubsection{Consistency with Examples}
\label{subsec:consthole}

\begin{figure}[t]\footnotesize

\infrule[Capturing group]
{(\sstate, w, p, \Gamma, \phi) \encode{} (\mathcal{S}, \mathcal{F})}
{((\sstate)_i, w, p, \Gamma, \phi) \encode{} (\bigcup_{ (p_i, \Gamma_i, \phi_{ci}) \in \mathcal{S} } (p_i, \Gamma_i[ i \mapsto w[p..p_i) ], \phi_{ci}), \mathcal{F}) }

\infrule[Backreference]
{ \text{Let $x$ = }\Gamma(i) \andalso x = w[p..p+|x|) }
{(\backslash i, w, p, \Gamma, \phi) \encode{} ( \{ (p+|x|, \Gamma, \phi) \} , \emptyset) }


\infrule[Positive lookahead]
{(\sstate, w, p, \Gamma, \phi) \encode{} (\mathcal{S}, \mathcal{F}) }
{(\text{(?=}\sstate\text{)}, w, p, \Gamma, \phi) \encode{} ( \{(p,\Gamma',\phi') \mid (\_,\Gamma', \phi') \in \mathcal{S}\} , \mathcal{F}) }

\infrule[Negative lookahead]
{(\sstate, w, p, \Gamma, \phi) \encode{} (\mathcal{S}, \mathcal{F})}
{(\text{(?!}\sstate\text{)}, w, p, \Gamma, \phi) \encode{}\\ ( \{ (p,\Gamma,\phi') \mid (\bot,\bot,\phi') \in \mathcal{F} \}, \{ (\bot,\bot,\phi') \mid (\_,\_,\phi') \in \mathcal{S} \}  ) } 

\infrule[Positive lookbehind]
{(x, w[p-|x|,p), 0, \Gamma, \phi) \encode{} (\mathcal{S}, \mathcal{F})}
{(\text{(?\textless=}x\text{)}, w, p, \Gamma, \phi) \encode{} (\{  (p,\Gamma, \phi') \mid (p', \Gamma', \phi')  \in \mathcal{S} \} , \mathcal{F}) }

\infrule[Negative lookbehind]
{(x, w[p-|x|,p), 0, \Gamma, \phi) \encode (\mathcal{S}, \mathcal{F}) }
{(\text{(?\textless!}x\text{)}, w, p, \Gamma, \phi) \encode{}\\ (\{ (p, \Gamma, \phi') \mid (\bot, \bot, \phi') \in \mathcal{F} \}, \{ (\bot,\bot,\phi') \mid (\_,\_,\phi') \in \mathcal{S} \}  ) }

\infrule[Hole]
{\text{$\hole$ is the $i$-th hole}}
{(\hole, w, p, \Gamma, \phi) \encode (\{ (p+1, \Gamma, \phi \land v_i^{w[p]}) \}, \{ (\bot, \bot, \phi \land \lnot v_i^{w[p]}) \}) }

\caption{Selected rules for generating consistency-with-examples constraints.}
\label{tab:encode} 
\end{figure}

To construct the constraint for ensuring the consistency with examples, we adapt and extend the approach proposed by \cite{10.1145/3360565} for constructing a similar constraint for pure regexes to real-world regexes.  
The main idea of \cite{10.1145/3360565} is to have a propositional variable $v_i^a$ for each $a \in \Sigma$ and $i$ that ranges over the number of holes in the given template $\sstate{}$ so that $v_i^a$ is true iff the set of characters $[C]$ to fill the $i$-th hole satisfies $a \in C$.
Then, the constraint is formulated to find an instantiation of $\sstate{}$ that satisfies (1) for each positive example, there is a run of the matching algorithm that accepts it, and (2) no run accepts a negative example.  

%
%
To this end, we define the function {\tt encode} which takes a template $\sstate$ and a string $w$.  It outputs the constraint $\phi_w$ that is satisfiable iff there exists an instantiation $r$ of $\sstate$ obtained by filling its holes with sets of characters such that $w \in L(r)$.  
%
The function {\tt encode} is defined by rules deriving judgements of the form $(\sstate, w, p, \Gamma, \phi) \encode (\mathcal{S}, \mathcal{F})$.  Here, $\phi$ accumulates the constraints asserted thus far, and $\mathcal{S}$ and $\mathcal{F}$ are sets of {\em constrained matching results} for successes and failures, respectively. A constrained matching result is a tuple $(p, \Gamma, \phi)$, where $p$ is a position, $\Gamma$ is a function that stores information about capturing groups, and $\phi$ is a constraint asserting the condition that must be satisfied for the corresponding matching to succeed or fail.  Matching results of the form $(\bot,\bot,\_)$ indicate matching failures.  Then, 
we define ${\tt encode}(\sstate, w) = \bigvee_{(|w|, \_, \phi) \in \mathcal{S} } \phi$ where
$(\sstate, w, 0, \emptyset, \mathit{true}) \encode (\mathcal{S}, \_)$.

Figure~\ref{tab:encode} shows the selected rules of $\encode$.  Here, the notation $\mathcal{M}[(p, \Gamma, \phi) \mapsto (p', \Gamma', \phi')]$, where $\mathcal{M}$ is either $\mathcal{S}$ or $\mathcal{F}$, denotes $\mathcal{M}$ but with $(p,\Gamma, \phi)$ replaced by $(p', \Gamma', \phi')$. For space, we only show the rules for handling the extended features and defer the full rules to the \supplementary{}.

Thanks to our rigorous formalization of the matching relation (cf.~Section~\ref{subsec:formalsemantics}), the constraint generation rules follow the corresponding rules of the matching relation and are almost straightforward.
The main difference is the rule (\textsc{Hole}) for processing holes.  The rule adds constraints to assert that the character $w[p]$ has to be included or not included in the set of characters that replaces the hole by conjoining $v_i^{w[p]}$ to the accumulated constraint $\phi$ for the success case, and conjoining $\neg v_i^{w[p]}$ to $\phi$ for the failure case.

Finally, the consistency-with-examples constraint for $\sstate$ is:
$\phi_c \triangleq \bigwedge_{w \in P} {\tt encode}(\sstate, w) \land \bigwedge_{w \in N}  \lnot~ {\tt encode}(\sstate, w)$.
\begin{example}
\normalfont
Consider the template $\sstate{} = (?!\hole)\hole bc$, the positive examples $P = \{ abc, cbc\}$, and the negative examples $N = \{bbc\}$.
For the positive examples, we have
\[
{\tt encode}(\sstate, abc) = \lnot v_0^a \land v_1^a \ \text{ and } \  {\tt encode}(\sstate, cbc) = \lnot v_0^c \land v_1^c.
\]
For the negative example, we have ${\tt encode}(\sstate, bbc) = \lnot v_0^b \land v_1^b$.
Therefore, $\phi_c = ( ( \lnot v_0^a \land v_1^a ) \land ( \lnot v_0^c \land v_1^c ) ) \land \lnot ( \lnot v_0^b \land v_1^b )$.

\end{example}

\subsubsection{Linear Time}
\label{subsec:oneprop}

\begin{algorithm}[t]
\caption{Generation of linear-time constraint}
\label{alg:linear}
\begin{algorithmic}[1]
\REQUIRE a template \sstate{}
\ENSURE a constraint $\phi_l$
\STATE $\sstate{} \leftarrow \removela(\sstate{})$
\STATE $\mathcal{A}$ $\leftarrow$ $\toENFA(\sstate{}^{[]})$ // $\mathcal{A} = (Q, \delta, q_0, q_n)$
\STATE $\phi_l \leftarrow {\it true}$
\STATE {\bf for each} $(q, [_i, q') \in \delta$ {\bf do}
\STATE \idt{} $L$ $\leftarrow$ {\follow}$(q)$
\STATE \idt{} {\bf if} $\rho_ia$, $\rho_ja$ $\in L$, where $\rho_i \neq \rho_j$ {\bf then}
\STATE \idt{} \idt{} {\bf return} {\it false}
\STATE \idt{} {\bf for each} $a \in L^{\natural}$ and $\hole_i \in L^{\natural}$ {\bf do}
\STATE \idt{} \idt{} $\phi_l \leftarrow \phi_l \land \lnot v_i^a$
\STATE \idt{} {\bf for each} $a \in \Sigma$ and $\hole_i, \hole_j \in L^{\natural}$ where $i \neq j$ {\bf do}
\STATE \idt{} \idt{} $\phi_l \leftarrow \phi_l \land (\neg v_i^a \vee \neg v_j^a)$
\STATE {\bf return} $\phi_l$
\end{algorithmic}
\end{algorithm}

Algorithm \ref{alg:linear} shows the construction of the linear-time constraint for enforcing \ltp{}.  It takes as input a template $\sstate$ and returns the linear-time constraint $\phi_l$.

%
The algorithm first removes lookarounds from the template by using $\removela$ defined in Definition~\ref{def:removela} (line 1).  Here, we extend $\removela$ to templates by treating each hole $\hole_i$ as a set of characters.  Condition (2) of \ltp{} which asserts repetition-freedom in lookarounds (cf. Definition~\ref{def:ltp}) is ensured by not placing repetitions in lookarounds of a template (cf.~\textbf{Expanding and adding holes to a template} in Section~\ref{subsec:overview}).

Next, the algorithm constructs a NFA for $\sstate{}^{[]}$ via the extended NFA translation defined in Definition~\ref{def:toenfa} (line 2). Then, for each open bracket $[_i$ in the NFA, the algorithm computes the set of paths $\follow(q)$ where $q$ is the source state of the (unique) $[_i$-labeled edge.  Here, we extend $\follow$ so that a hole $\hole_i$ is treated as the set of characters $[\hole_i]$ (cf.~Section~\ref{subsec:lineartime}).  

We then check if there are multiple brackets-only routes from $[_i$ that reach a same character (line 6).  If the check passes, then |$\computePath(\removela(r), [_i, a)| \geq 2$ for any regex $r$ obtainable from $\sstate$ violating condition (1) of \ltp{}, and we safely reject $\sstate$ by returning the unsatisfiable formula {\it false}.  

Otherwise, we proceed to add two types of constraints in lines 8-11.  The constraints of the first type added in lines 8-9 assert that, if some character $a \in \Sigma$ and a hole $\hole_i$ are both reachable from $[_i$ by bracketing-only paths, then the hole must not be filled with a set of characters that contains $a$.  Here, $L^{\natural} = \{\alpha \mid \rho\alpha \in L\}$  (cf.~Section~\ref{subsec:lineartime}).  The constraints of the second type added in lines 10-11 assert that, if there are two different holes $\hole_i$ and $\hole_j$ reachable from $[_i$ by bracketing-only paths, then for any character $a \in \Sigma$, at most one of the hole can be filled with a set of characters that contains $a$.  It is easy to see that condition (1) of \ltp{} is satisfied iff these constraints are satisfied for all $[_i$.
Finally, the algorithm returns the resulting constraint $\phi_l$ (line 12).




%

\begin{figure}[t]
\begin{center}
\includegraphics[width=85mm]{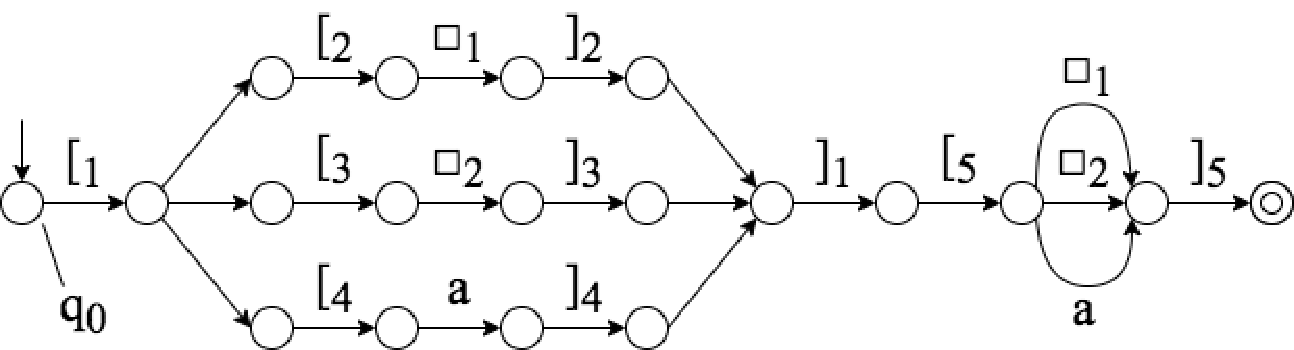}
\caption{Simplified version of the extended NFA translation of $[_1([_2 \hole_1 ]_2 | [_3 \hole_2 ]_3 | [_4 \texttt{a} ]_4 )_1 ]_1 [_5 \backslash 1 ]_5$.}
\label{fig:constexa}
\end{center}
\end{figure}
\begin{example}
\normalfont
Let us consider running the algorithm on the template $\sstate{} = (\hole_1 | \hole_2 | \texttt{a})_1 \backslash 1 (\mbox{?!}\texttt{a})$.
The algorithm first removes lookarounds in the template (line 1), and thus the template becomes $(\hole_1 | \hole_2 | \texttt{a})_1 \backslash 1$.
Next, the algorithm applies the extended NFA translation to $\sstate{}^{[]}$ (line 2).
Here, $\sstate{}^{[]}$ is $[_1([_2 \hole_1 ]_2 | [_3 \hole_2 ]_3 | [_4 \texttt{a} ]_4 )_1\:]_1\: [_5 \backslash 1 ]_5$.
For brevity, we omit the some redundant brackets for sequences and unions.
Figure~\ref{fig:constexa} shows the obtained NFA.
Then, the algorithm constructs the constraints (lines 4-11).
Let us consider the case of $q_{0}$.
In this case, $\follow(q_{0}) = \{ [_1 [_2 \hole_1, [_1 [_3 \hole_2, [_1 [_4 \texttt{a}\}$.
Line 9 adds the constraint $\lnot v_1^\texttt{a} \land \lnot v_2^\texttt{a}$ and line 11 adds the constraint $\bigwedge_{a \in \Sigma} (\neg v_1^a \vee \neg v_2^a)$ to $\phi_l$.
\end{example}

\subsection{Optimization}
\label{subsec:optim}
We show an optimization to the algorithm.  When adding new holes to a template at line 9 of Algorithm~\ref{alg:repair}, we select the sets of characters to be replaced by holes as follows.  We analyze the result of the extended NFA translation (which is done anyway for the linear-time constraint) to identify the sets of characters that violate the \ltp{} condition, and replace only those with holes.  This has the effect of reducing the search space by focusing the synthesis to the parts that contribute to vulnerability.



For example, from the template \texttt{<s$\hole{}$an$\any^*$>}, without the optimization, we may generate up to $2^6$ templates by replacing the sets of characters by holes.
But with the optimization, we only generate one template \texttt{<s$\hole{}$an$\hole{}^*\hole{}$} because $\any$ and \texttt{>} are the only sets of characters that violate the \ltp{} condition.


\section{Implementation and Evaluation}
\label{sec:eval}

In this section, we present the results of our evaluation.
We evaluate the performance of \tool{} by answering the following questions.
\begin{description}
\item[RQ1] Can \tool{} repair vulnerable regexes efficiently?
\item[RQ2] Can \tool{} find high-quality regexes?
\item[RQ3] What is the effect of the optimization?
\end{description}

For the first question, we measure the time taken to repair vulnerable regexes on a real-world data set.
For the second question, we measure the quality of repaired regexes using the metrics also used in~\cite{10.1145/3360565}.
For the last question, we compare the running times of \tool{} and \tool{} with the optimization described in Section~\ref{subsec:optim}.
Henceforth, we refer to \tool{} with the optimization as \tool{}-o, and use \tool{}-h to denote the hybrid of \tool{} and \tool{}-o that returns the regex returned by the faster of the two.

Finally, we present a comparison of our tool~\tool{} with the other state-of-the-art tools in Section~\ref{subsec:compare}. We compared \tool{} with three state-of-the-art tools AlphaRegex~\cite{10.1145/3093335.2993244}, RFixer~\cite{10.1145/3360565}, and FlashRegex~\cite{FlashRegex}. AlphaRegex only supports synthesizing a regex, while RFixer and FlashRegex support both synthesizing and repairing a regex.

 
\subsection{Experimental Setup}
\label{subsec:expsetup}
We have implemented \tool{} in Java.  We use Z3~\cite{10.5555/1792734.1792766} as the SMT solver.
All experiments were performed on a machine with Intel(R) Xeon(R) Gold 6254 CPU @ 3.10GHz.

{\flushleft\bf Benchmark.}
We used {\it Ecosystem ReDoS data set} collected by Davis et al.~\cite{Davis:2018:IRE:3236024.3236027}, which contains real-world regexes in Node.js (JavaScript) and Python core libraries.
The data set contains 13,670 regexes that contain real-world extensions (i.e., lookarounds or backreferences).
Initially, the regexes are not classified whether they are vulnerable or not.
Thus, we contacted the authors of \cite{Davis:2018:IRE:3236024.3236027} to obtain the subset that they classified as vulnerable.
As a \tchanged{consequent}{consequence}, the data set contains 13,591 regexes that contain real-world extensions and are unknown whether they are vulnerable or not, and 79 regexes that contain real-world extensions and are vulnerable.
Due to the size, for the former, we selected 100 of them randomly. For the latter, we selected all of them.
The average and maximum sizes of the regexes (measured as number of AST nodes) are 32.1 and 383, respectively.

We note that there are no known sound-and-complete ReDoS vulnerability detection methods for real-world regexes (in fact, even whether such a detection is possible is an open question).
The 79 regexes that are classified as vulnerable are manually classified as so by Davis et al.
%
We have chosen this data set because its regexes represent real use cases and are also considered to be vulnerable. 


{\flushleft\bf Sampling Examples.}
Since the data set of \cite{Davis:2018:IRE:3236024.3236027} do not come with examples, we prepared the examples by ourselves.  
Many of them were made manually, but some were generated automatically, due to the large sizes of the regexes, by the following input generation technique that is inspired by that of \cite{10.1145/3236024.3236072} for pure regexes.

We first convert the given regex to a backreference-free regex by replacing each capturing group $(r)_i$ and backreference $\backslash i$ by fresh symbols $\alpha_i$ and $\beta_i$, respectively.  The resulting pure regex (lookarounds can be eliminated for backreference-free regexes~\cite{DBLP:journals/jip/MiyazakiM19}) is converted to a DFA.  We enumerate the accepting paths of the DFA so that each edge appears in at least one path, with the requirement that an edge $\beta_i$ can only be taken if the corresponding edge $\alpha_i$ was taken before in the path.  Each path is turned into a set of positive examples by replacing each $\alpha_i$ and $\beta_i$ by a positive example of the regex $r$ where $(r)_i$ is the capturing group (positive examples of $r$ are generated by recursively applying this process).  Negative examples are generated similarly by considering the rejecting paths of the DFA. 

Finally, we used at most 5 positive and negative examples each. We note that, for usability, PBE should only use relatively small numbers of examples.

{\flushleft\bf Consistency with Examples.}
By construction, \tool{} is guaranteed to only generate regexes that are consistent with the given examples. 
We have also validated that all regexes that \tool{} generated in the experiment were indeed consistent with the given examples by running the Java's regex library {\tt util.regex}.

{\flushleft\bf ReDoS Invulnerability.}
By construction, \tool{} is guaranteed to only generate regexes that satisfy \ltp{} and hence ReDoS invulnerable.
We have also validated that all regexes that \tool{} generated in the experiment indeed satisfied \ltp{}.  Note that whether
a regex satisfies \ltp{} can be easily checked by analyzing the extended NFA translation of the regex (cf.~Section~\ref{subsec:lineartime}).


\subsection{RQ1: Efficiency}
\begin{figure}[t]
 \begin{minipage}[b]{0.49\linewidth}
\includegraphics[width=\hsize]{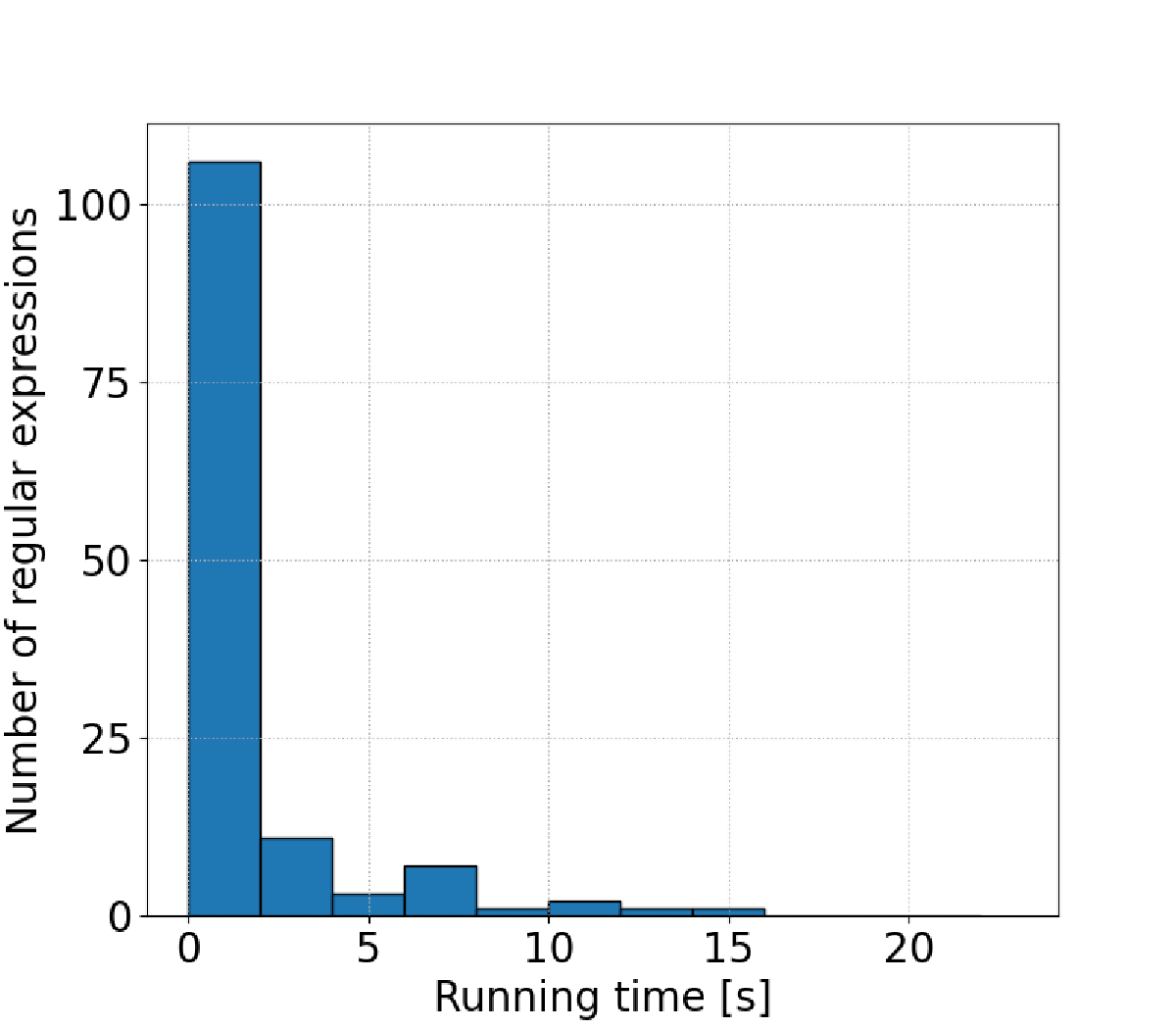} 
  \subcaption{Running times of \tool{}}\label{fig:rq1_normal}
 \end{minipage}
 \begin{minipage}[b]{0.49\linewidth}
  \includegraphics[width=\hsize]{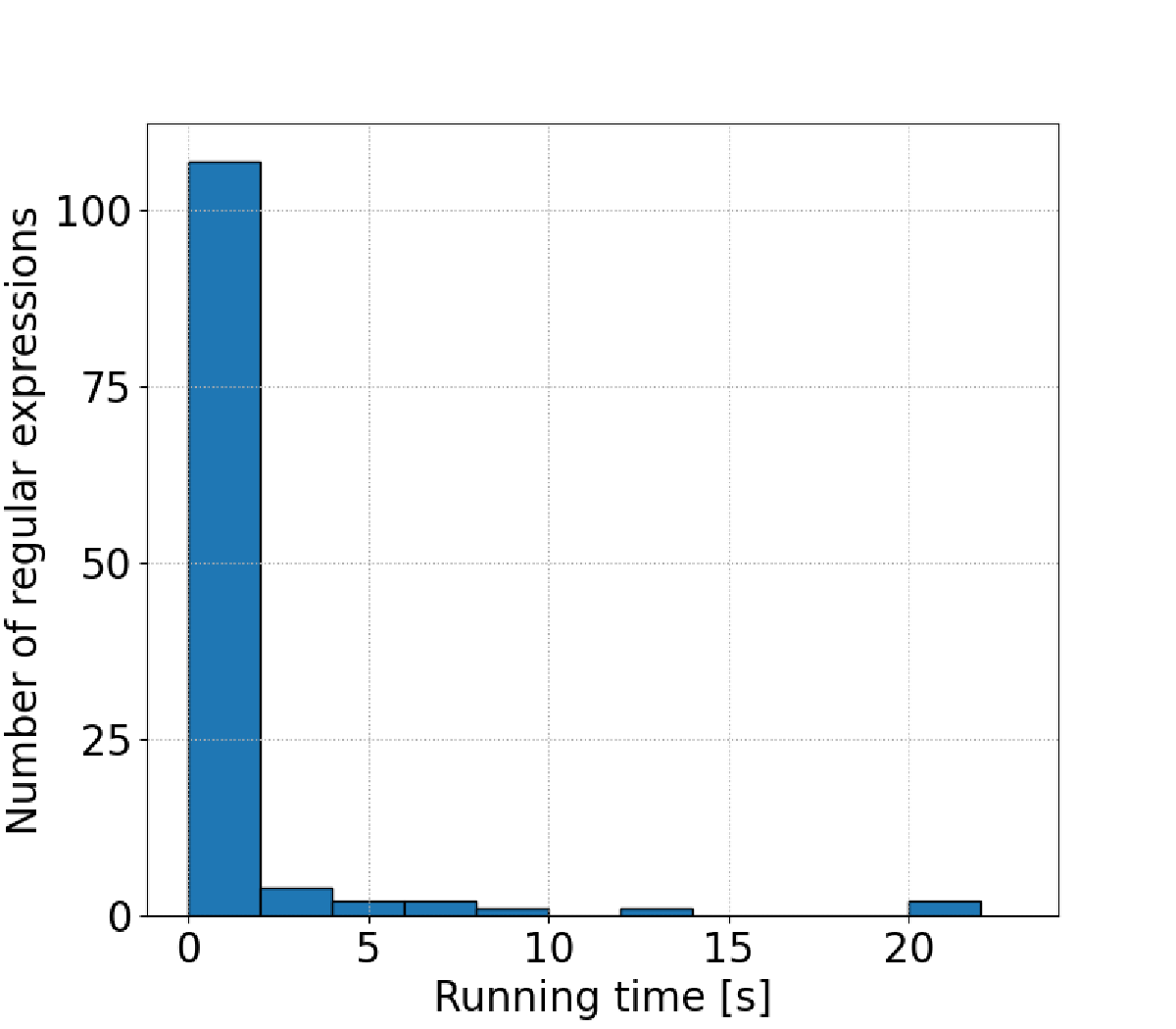} 
  \subcaption{Running times of \tool{}-o}\label{fig:rq1_optim}
 \end{minipage} 
 \caption{Results of the repairs.}\label{fig:result_summary}
\end{figure}


\subsubsection{Performance}
To evaluate the performance, we ran \tool{} and the variants with a timeout of 30 seconds.
We chose 30 seconds because the improvement by setting the timeout to more than 30 seconds was little.
The table below summarizes the results.
The columns {\bf Solved} and {\bf Average} show that the number of test cases which were repaired within the timeout range and the average running time, respectively.
Additionally, Figure \ref{fig:rq1_normal} and \ref{fig:rq1_optim} summarize the running times.

\begin{wraptable}{r}{55mm}
\setlength\tabcolsep{1pt}
\begin{tabular}{lr}
  \begin{tabular}{lcc}
          & {\bf Solved}(179) & {\bf Average}(s) \\ \hline
    \tool{}   & 132 & 1.54 \\
    \tool{}-o & 119 & 1.08 \\ 
    \tool{}-h & 147 & 0.97 \\
      \end{tabular}
\end{tabular}
\vspace{-0.3cm}
\end{wraptable}
In total, \tool{}, \tool{}-o, and \tool{}-h repaired 73.7\%, 66.5\%, and 82.1\% of regexes, respectively.
More than 82.3\% of regexes were repaired within 1 second.
On the other hand, we observed that the tools could not repair 17.9\% of regexes within the time limit.  Our inspection showed that the tools struggled on repairs that require large changes from the original.\tchanged{}{ Such repairs may need to explore a large space of possible regexes.
An example of such failure cases is a regex that contains a concatenation of many vulnerable sub-regexes, e.g., \texttt{($\any^*$[,])$^*$[,]$^+$[ ]$^+$(['"]$?$)[ a]$^*\backslash$2$\cdots$}, where $\cdots$ is a further concatenations of vulnerable sub-regexes.}
The finding agrees with that of \cite{10.1145/3360565} who have reported \tchanged{the same issue for their repair method}{that their method also struggled on repairs with large changes,} and whose techniques are adapted to our method (cf.~Section~\ref{sec:algo}).
In summary, {\em \tool{} can repair vulnerable regexes that contain real-world extensions efficiently}.

\subsubsection{Scalability}
\label{sec:scalability}
\begin{figure}[t]
 \begin{minipage}[b]{0.49\linewidth}
  \centering
  \includegraphics[width=\hsize]{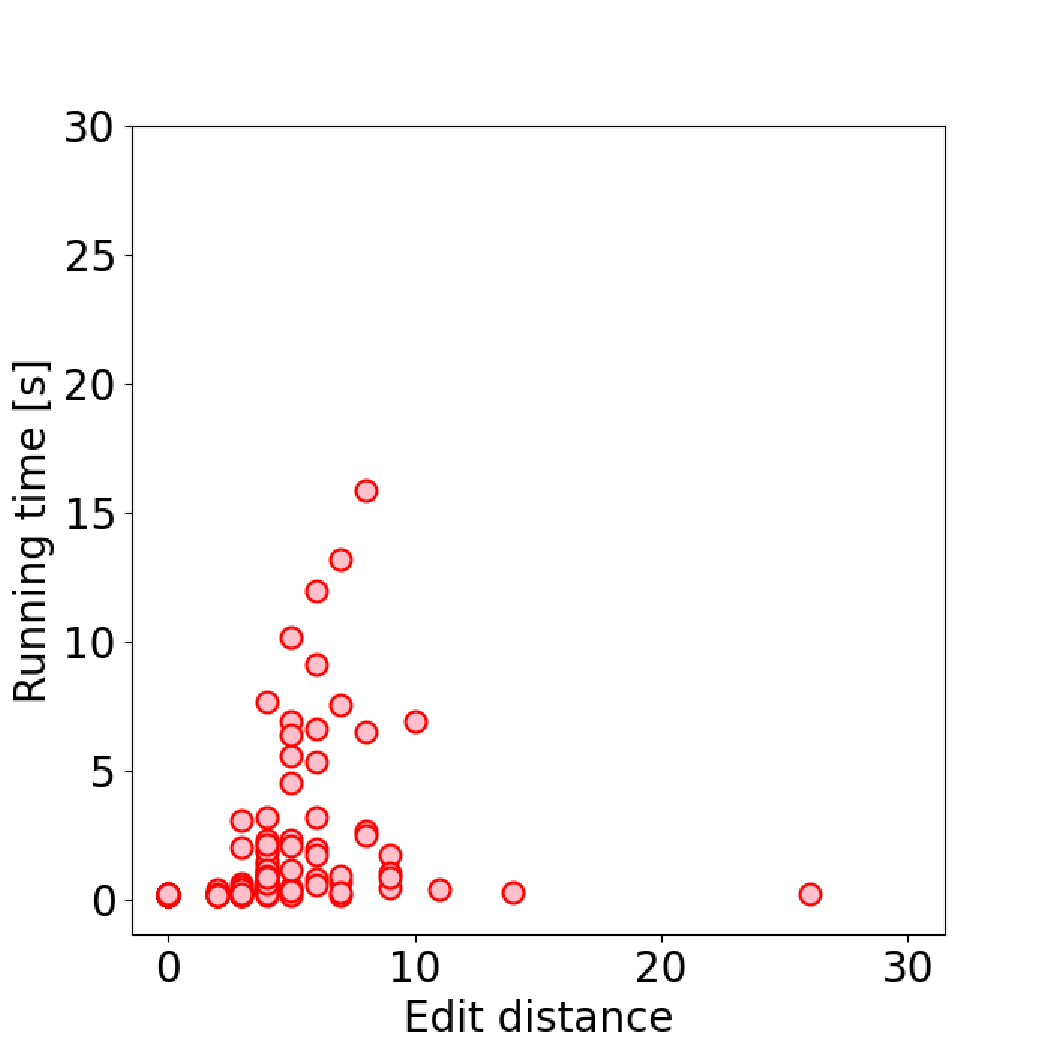}  
  \subcaption{\tool{}}\label{fig:scale_edit_REMEDY}
 \end{minipage}
 \begin{minipage}[b]{0.49\linewidth}
  \centering
  \includegraphics[width=\hsize]{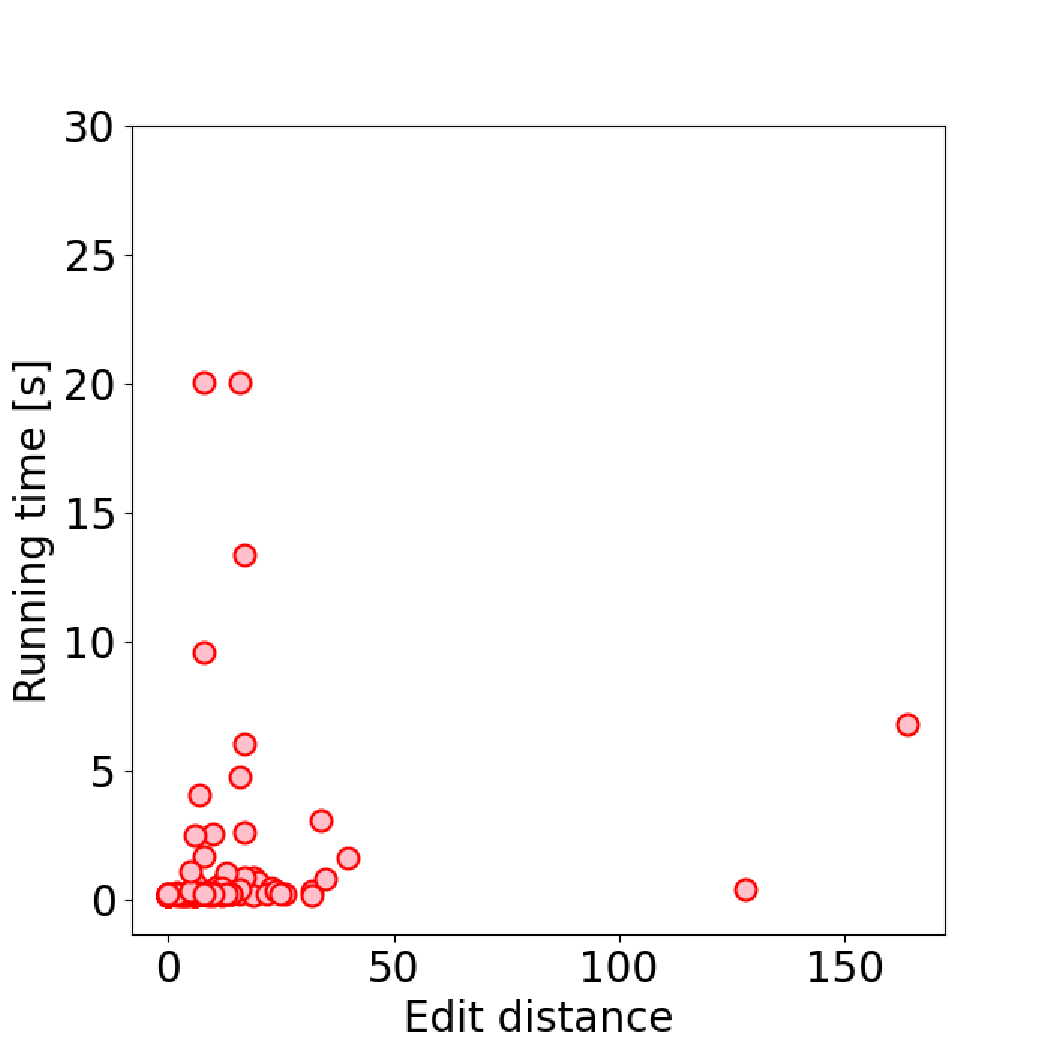} 
  \subcaption{\tool{}-o}\label{fig:scable_edit_REMEDY_o}
 \end{minipage} 
 \caption{Scalability with respect to edit distances.}\label{fig:scale_edit}
\end{figure}
Based on the finding, we plot the running times of \tool{} over the edit distances from the input regexes to their repair results.  
As Figure \ref{fig:scale_edit_REMEDY} shows, in the case of \tool{}, we observed a general correlation between the running times and the edit distances: large edit distances require long running times. This observation affirms our initial findings that the size of edit distance affects the running time.

We also observed that, in \tchanged{one case}{some cases}, \tool{} finished repairing with a large edit distance and a short running time.  To explain the behavior, we use as example the regex below that is derived from the actual case.
\[
\texttt{(?=$(\any^*)_1$[ ]$([0$-$9][$:$]\backslash1[$:$][$:$])^*$)$\any^*$}
\]
The regex contains a lookahead with repetitions and backreferences thus violating the condition (2) of \ltp{}.  \tool{} immediately detects the violation and replaces the repetitions and backreferences with holes.  As result, \tool{} reaches the template \texttt{(?=$(\hole{})_1$[ ]$\hole{}$)$\any^*$} in one step.  Note that the template replaced quite large sub-expressions with holes and is of a large edit distance.  This substantially reduced the search space, and thus the short running time was achieved even with the large edit distance.

On the other hand, as shown in Figure~\ref{fig:scable_edit_REMEDY_o}, we observed less correlation between the running times and the edit distances for \tool{}-o.  Note that the scale of the edit distance axis is significantly wider than that of Figure~\ref{fig:scale_edit_REMEDY}.  The observation also coincides with our initial findings because \tool{}-o uses the optimization described in Section~\ref{subsec:optim} that can increase the edit distance in a small number of steps\tchanged{.}{, namely all sets of characters violating condition (1) of \ltp{} are immediately replaced by holes.}
\texcomment{
\tchanged{}{ For example, the points at edit distance $\sim$130 in Figure \ref{fig:scable_edit_REMEDY_o} are caused by the optimization regarding the condition (1) of \ltp{}. }
}

\begin{wrapfigure}{r}{50mm}
  \centering
    \includegraphics[width=\hsize]{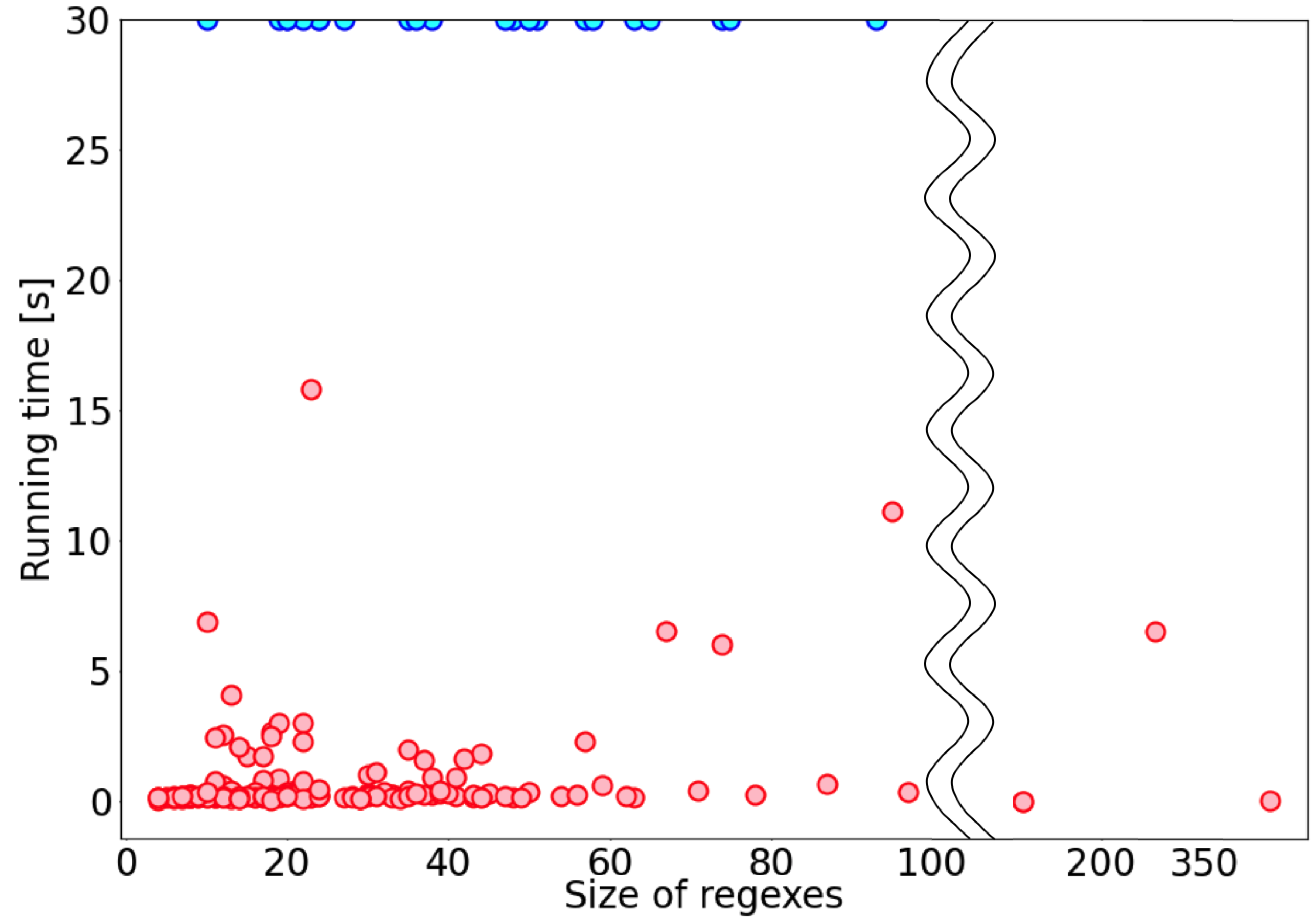}
     \caption{Scalability wrt. regex sizes.}
    \label{fig:scale} 
\end{wrapfigure} 
Additionally, to understand how our tool scales as the size of a regex increases, we plot the running time of \tool{}-h over the size of the regex (measured as number of AST nodes).
Figure \ref{fig:scale} shows the result.
The points on the border (colored in blue) indicate that \tool{}-h could not repair the regex within the time limit.
Note that the figure is truncated to omit redundant space where no points appear.

One can observe that, except for the regexes which led to timeout, \tool{} could repair almost all regexes within 1 second regardless of their size.
Additionally, we inspected some of the regexes which require more than 1 seconds to repair and confirmed that they require large changes from the original.
That is, the impact of the size of regexes on the implementation is little, while the size of edit distance affects the implementation.

In summary, {\em the performance of \tool{} scales with the size of regexes}. Additionally, {\em there is a correlation between the running times and the edit distances of \tool{}. }

\subsection{RQ2: Quality}
\label{subsec:rq2}
\begin{figure}[t]
 \begin{minipage}[b]{0.49\linewidth}
  \centering
    \includegraphics[width=\hsize]{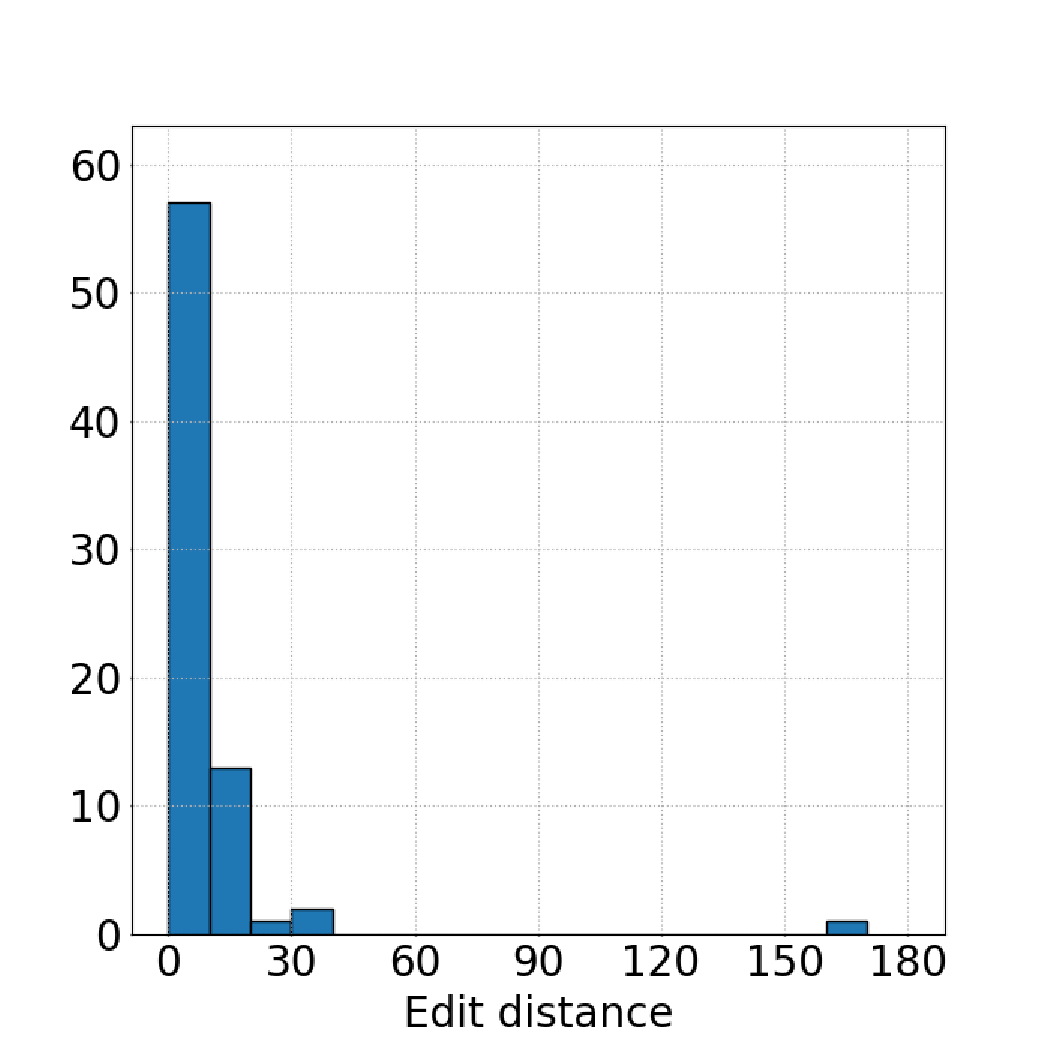}
     \subcaption{Edit distances.}
    \label{fig:histoedit}   
 \end{minipage}
 \begin{minipage}[b]{0.49\linewidth}
  \centering
  \includegraphics[width=\hsize]{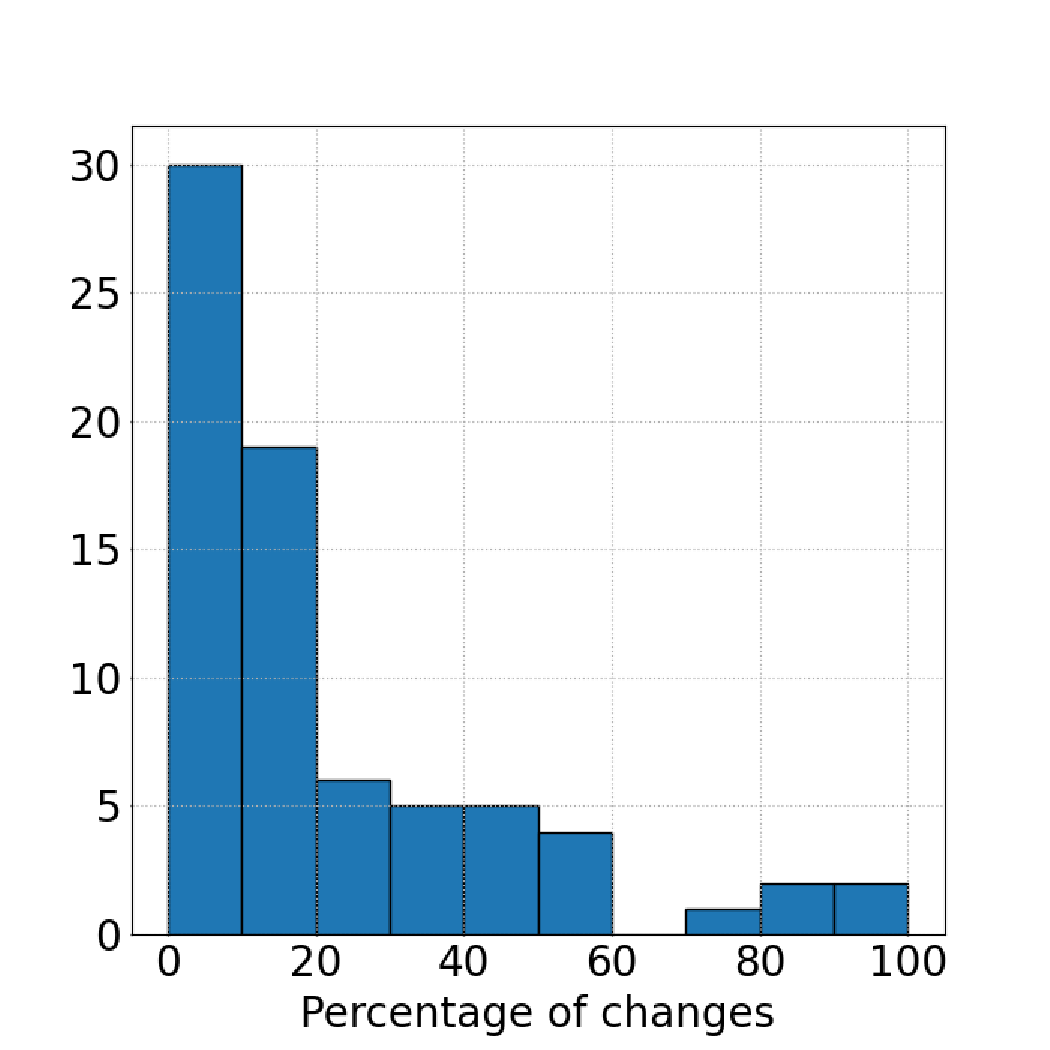}
  \subcaption{Percentages of changes.}\label{fig:histoeditquality}
 \end{minipage} 
 \caption{Histograms for repair quality.}\label{fig:histograms_of_edits}
\end{figure}

As mentioned by \cite{10.1145/3360565}, repairs that are similar to the original ones are often considered good in PBE because they are similar to what the user intended.
Therefore, to evaluate the quality of repaired regexes objectively, we measure the similarity to the original regex.
%
%
A large change indicates low quality as such repairs may be far from what the user intended.
Figure \ref{fig:histoedit} shows a histogram plotting the number of regexes against the edit distances to their repair results by \tool{}-h. Most of the regexes were repaired within the small edit distances, with about 81\% repaired within edit distance 12.

We also measure the ratio of changes, i.e., the size of the regex portion changed by its repair divided by the size of the entire regex.  
%
Figure \ref{fig:histoeditquality} shows a histogram plotting the number of regexes against their ratios of changes.
We observe that most repairs are close to the original regexes, with the average ratio of change being 24.3\%.

\tchanged{}{
We discuss some typical cases of repairs that we observed in our experiments.  For example, the data set contained vulnerable regexes that use positive lookaheads to assert an appearance of some keyword, e.g., \texttt{$\any^*$(?=[ ]$^*$[;])$\any^*$}. For this, \tool{} returns the repaired invulnerable regex \texttt{[\textasciicircum;]$^*$[;]$\any^*$}, which is semantically equivalent to the original, with the edit distance of 7.  We also refer to the XML example from Section~\ref{sec:example} as an exemplar repair case that we observed in our experiments.}
In summary, \tchanged{{\em \tool{} can find high-quality regexes}}{{\em \tool{} can produce repaired regexes that have high-similarity, and therefore of high-quality}}.  


\subsection{RQ3: Effect of the Optimization}
\begin{wrapfigure}{r}{45mm}
\vspace{-1em}
  \centering
    \includegraphics[width=\hsize]{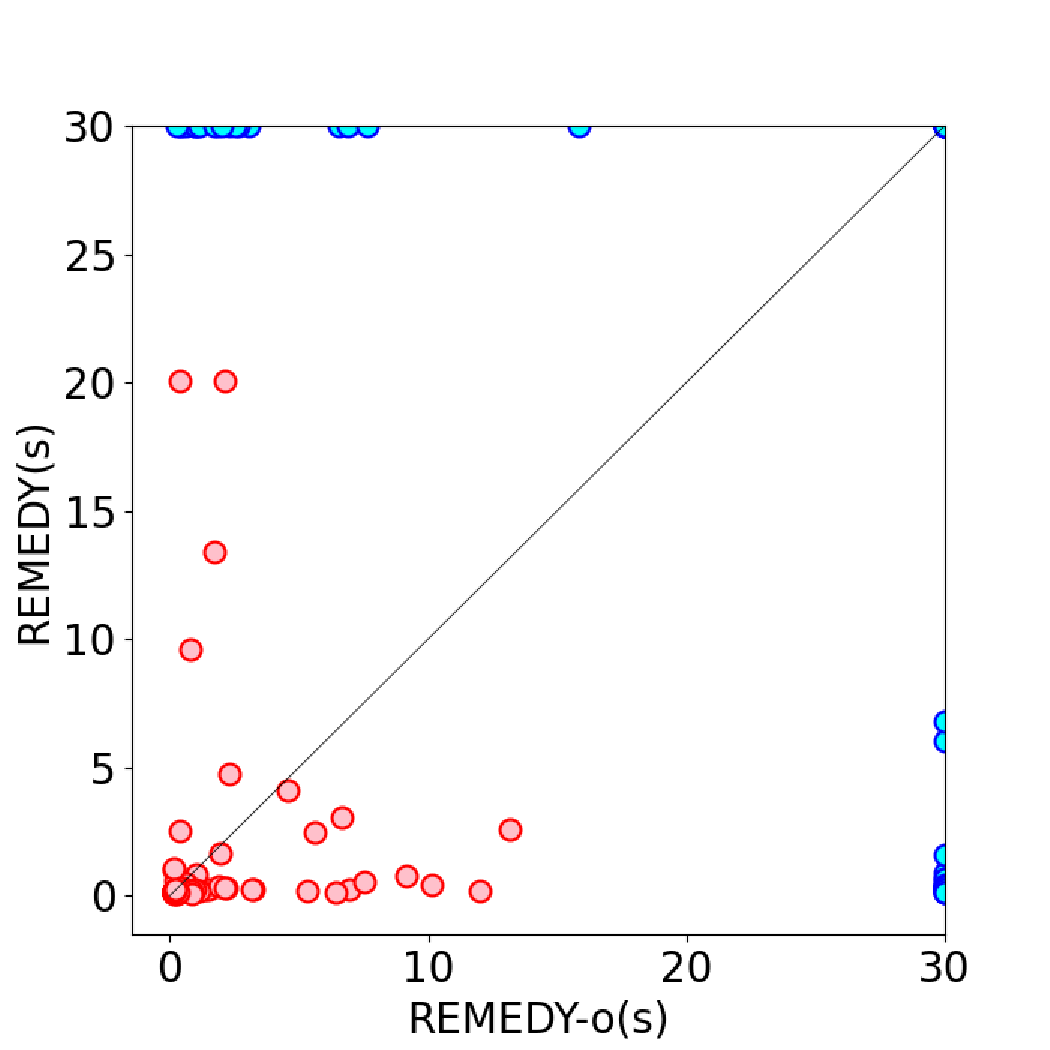}
     \caption{Optimization effect.}
    \label{fig:impact} 
\vspace{-0.5em}
\end{wrapfigure}
We evaluate the effectiveness of the optimization described in Section \ref{subsec:optim}.
The comparison of the running times of \tool{} and \tool{}-o are shown in Figure \ref{fig:impact}. 
We note that, in 23 cases, \tool{}-o solved the instance within the time limit while \tool{} could not, and
conversely in 19 cases, \tool{} solved the instance within the time limit while \tool{}-o could not.

We have observed that \tool{}-o often outperformed \tool{} for regexes that violate the \ltp{} condition at many places.
For example, for the regex 
\[\small
\texttt{(?=[\textasciicircum,])$\any^*$,(?=[\textasciicircum,])$\any^*$,(?=[\textasciicircum,])$\any^*$,(?=[\textasciicircum,])$\any^*$,$\any^+$}
\] 
the desired template is one in which the first four any character, i.e., $\any$, is replaced with a hole.
\tool{} reaches such a template only after trying $2^4-1$ many other templates, whereas
\tool{}-o reaches it immediately.  Conversely, in cases where \tool{} performed better, we have observed that 
the repair benefits from templates that replace the non-\ltp{}-violating parts of the regex with holes.  Since the optimization prevents \textsf{addHoles} from replacing such parts with holes, it can negatively affect the performance in such cases.

In summary, {\em the optimization helps \tool{} to repair regexes that violate the \ltp{} condition at many parts, while it negatively affects cases where non-\ltp{}-violating parts should be repaired}.  Thus, {\em running both \tool{} and \tool{}-o, i.e., \tool{}-h, achieves better performance than running one of them alone}.


\subsection{Comparison to Other State-of-the-art Tools}
\label{subsec:compare}

\begin{table*}[t]
\caption{A comparison of current state-of-the-art tools. \Checkmark~and \xmark~indicate that the tool has and does not have the characteristic, respectively. {\bf Binary Alphabet} and {\bf Multiple Alphabet} indicate that the tool can synthesize a regex over binary and multiple alphabets, respectively. {\bf Correctness Guarantees} indicates that the tool guarantees that synthesized regexes are consistent with all examples.  {\bf Invulnerability Guarantees} indicates that the tool guarantees that synthesized regexes are not ReDoS vulnerable. {\bf Real-world Extensions} indicates that the tool can support real-world extensions.}
\label{tab:compare}
\centering
\begin{tabular}{l|c|c|c|c|c}
\toprule
{\bf Tool} & {\bf Binary Alphabet} & {\bf Multiple Alphabet} & {\bf Correctness Guarantees} & {\bf Invulnerability Guarantees} & {\bf Real-world Extensions} \\
\midrule
\tool{} & \Checkmark & \Checkmark & \Checkmark & \Checkmark & \Checkmark \\
\rowcolor[gray]{0.8}AlphaRegex~\cite{10.1145/3093335.2993244} & \Checkmark & \xmark & \Checkmark & \xmark & \xmark \\
RFixer~\cite{10.1145/3360565} & \Checkmark & \Checkmark & \Checkmark & \xmark & \xmark \\
\rowcolor[gray]{0.8}FlashRegex~\cite{FlashRegex} & \Checkmark & \Checkmark & \Checkmark & \xmark & \xmark \\
\bottomrule
\end{tabular}
\end{table*}

Table \ref{tab:compare} summarizes the characteristics of different PBE tools for generating regexes, and compare them with our tool~\tool{}.
The {\bf Invulnerability Guarantees} column shows that \tool{} is the only one to guarantee the invulnerability.
Indeed, for AlphaRegex and RFixer, AlphaRegex generates vulnerable regexes, e.g., some regexes shown in Table 3 of \cite{10.1145/3093335.2993244} are vulnerable, and RFixer often generates vulnerable regexes as reported by Li et al.~\cite{FlashRegex}.
For FlashRegex, we could not confirm whether FlashRegex actually generates vulnerable regexes or not because FlashRegex is not publicly available (only the dataset is available from the GitHub repository). Additionally, we have contacted the authors, but the implementation was not available.
However, FlashRegex claims to generate invulnerable regexes by only generating deterministic (i.e., 1-unambiguous) regexes, which unfortunately is insufficient for guaranteeing invulnerability as we have shown in Section~\ref{subsec:lineartime}.

The {\bf Real-world Extensions} column shows that \tool{} is the only one to support real-world extensions.
The other state-of-the-art tools, i.e., AlphaRegex, RFixer, and FlashRegex, only support pure regexes and supporting real-world extensions is out of scope for their work as mentioned in their papers.
Additionally, supporting them would require a substantial overhaul as described in Sections \ref{sec:regex}, \ref{sec:problem}, and \ref{sec:algo}.

In summary, {\em \tool{} improves the other state-of-the-art tools from a theoretical point of view, and is the only one to support all characteristics}.
We emphasize that all the other state-of-the-art tools can repair none of the regexes used in the evaluation.

Furthermore, we have compared \tool{} against the DFA-based approach of \cite{10.1145/3129416.3129440}.  Their approach is not PBE but claims to produce an invulnerable pure regex that is semantically equivalent to the given pure regex. 
We performed an experiment by using 100 pure regexes randomly selected from the data set of \cite{Davis:2018:IRE:3236024.3236027}, and compare the size of the regex repaired by the DFA-based approach and \tool{}.

We observed that (1) 100/100 of the regexes repaired by \tool{} are more concise than those of the DFA-based approach, and (2) the size of the regex repaired by the DFA-based approach is 37.3 times larger than that of \tool{} on average.
We observe that this is partly because \tool{} can use real-world extensions (even for repairing pure regexes), and also because the DFA-based approach ensures semantic equivalence, which is often undesirable (cf. Section \ref{sec:limit}).  Note that a semantics-preserving DFA conversion can generate exponentially large DFAs.
For example, for the regex \texttt{$\any^*\any^*$=}, which is vulnerable,
\tool{} returns the repaired invulnerable regex \texttt{$\any\any^*$(?<=[=])}.  On the other hand, the one produced by the DFA-based approach is
\[
\texttt{[\textasciicircum=]$^*$[=]([=]|[\textasciicircum=][\textasciicircum=]$^*$[=])$^*$}
\]
As another example, for the vulnerable regex 
\[
\texttt{<span[\textasciicircum>]$^*$}\texttt{font}\texttt{-}\texttt{s}\texttt{t}\texttt{y}\texttt{l}\texttt{e}\texttt{:italic[\textasciicircum>]$^*$>}
\]
\tool{} returns the repaired invulnerable regex 
\[
\begin{array}{l}
\texttt{<span([\textasciicircum".1-8B-Y$\backslash$[$\backslash\backslash\backslash$]\textasciicircum b-dfh-y]$^*$)}\\
\hspace{5em} \texttt{font}\texttt{-}\texttt{s}\texttt{t}\texttt{y}\texttt{l}\texttt{e}\texttt{:}\texttt{i}\texttt{t}\texttt{a\-}\texttt{l}\texttt{i}\texttt{c}\texttt{([\textasciicircum>]$^*$)>}
\end{array}
\]
while the one returned by the DFA-based approach is of size 73,433,094.
In summary, compared to the DFA-based approach, {\em \tool{} can find simpler and more understandable regexes}.

\texcomment{
\tchanged{}{
Finally, we compare \tool{} with global ReDoS avoidance mechanisms.
Some regex engines incorporate the functionalities to limit the number of backtrackings such as the {\tt pcre.backtrack\_limit} option of PHP or set a timeout limit such as the {\tt Regex.MatchTimeout} property of .NET to defend against ReDoS threats. 
However, the functionalities suffer from the problem that determining a proper limit is non-trivial~\cite{revealear}, which sometimes lead to program bugs or fail to avoid ReDoS, and are not always an option.

In summary, {\em \tool{} offers a complementary solution to ReDoS by repairing vulnerable regexes to secure ones that do not need externally-forced backtracking limits, and therefore are free of these bugs.} Additionally, {\em \tool{} offers a solution that is language/library independent.}

}
}

\subsection{Availability}
Our tool is available in \cite{remedy}.

\section{Limitations and Future Work}
\label{sec:limit}

We discuss some limitations of our approach and directions for future work.
The first limitation is that we do not consider {\em extraction} of captured strings.  This is a common limitation in regex repair and synthesis and many other recent works also do not support extraction~\cite{10.1145/3360565,10.1145/3093335.2993244,DBLP:journals/corr/abs-1908-03316,FlashRegex}.  

Extraction is especially problematic for real-world regexes as which string is captured in a lookahead is regex engine dependent.\footnote{It can even cause differences in the matching results in the rare cases where strings captured in lookaheads are backreferenced.  For example, matching $\texttt{(?=(a$^*$)$_1^+$)$\backslash$1a}$ with $\texttt{a}$ succeeds in Python's {\em re} and PCRE, but
fails in ECMAScript~\cite{10.1145/3338906.3338909}.}  A recent work has proposed an approach to cope with the issue in the context of symbolic execution~\cite{10.1145/3314221.3314645} that involves executing an actual regex engine.  But such an approach would be less ideal for repairs where we want to generate a regex that is correct and invulnerable in all contexts.  We leave as future work to investigate the support for extraction.  
It is important to note that our formal definition of vulnerability considers all possible captures that can happen in a lookahead, and thus our approach is regex-engine-{\em in}dependently sound with respect to invulnerability.

The second limitation is the lack of support for {\em semantic equivalence}.  As in other PBE methods, we consider the use case where the given regex is incorrect or only partly built.  As argued by others~\cite{FlashRegex,Davis:2018:IRE:3236024.3236027,Shen:2018:RCR:3238147.3238159}, often, semantic equivalence is too strong to use in practice and PBE is better at reflecting users' intentions.  But in future work, we would like to also support the case where the user is interested in only repairing vulnerability (e.g., because the regex is built correct by using some PBE method).  However, whether a regex can be repaired to be invulnerable while preserving its semantics in general is an open problem.  At least for real-world regexes, there are some reasons to doubt the possibility: semantic equivalence of real-world regexes is undecidable~\cite{backreferenceisundeci} and regexes with backreferences are not determinizable in general~\cite{SCHMID20161}.                            

\section{Related work}
\label{sec:related}

As remarked in Section~\ref{sec:intro}, there has been substantial work on PBE methods for synthesizing and repairing regexes~\cite{Alquezar94incrementalgrammatical,10.1145/3093335.2993244, 6994453, 7374717, 10.1145/3360565, DBLP:journals/corr/abs-1908-03316, FlashRegex}.
However, the existing methods do not support the real-world features such as lookarounds and backreferences.  Furthermore, with the
exception of \cite{FlashRegex} discussed below, the existing methods are not designed with resilience to ReDoS in mind and may generate vulnerable 
regexes.

A recent work by Li et al.~\cite{FlashRegex} proposes a PBE regex synthesis and repair method that addresses vulnerability.  Their method guarantees that the generated regex is deterministic (i.e., $1$-unambiguous)~\cite{BRUGGEMANNKLEIN1998182,10.1007/3-540-57273-2_45}.  However, as we have shown in Section~\ref{subsec:lineartime}, $1$-unambiguity is insufficient for invulnerability. Therefore, their method does not guarantee the invulnerability of the returned regexes.  Also, their method only synthesizes and repairs pure regexes and does not support the real-world extensions.   By contrast, our work supports real-world regexes and also formally guarantees the invulnerability of the synthesized regexes.

While not PBE, the work by van der Merwe et al.~\cite{10.1145/3129416.3129440} proposes a technique based on DFA conversion and insertion of positive lookaheads to convert a vulnerable regex into an invulnerable one.  However, they only consider the fragment with the positive lookahead extension.  Also, as discussed in Section~\ref{subsec:compare}, the DFA-based approach can produce complex regexes that are hard to understand.  In a similar vein, Cody-Kenny et al.~\cite{10.1145/3071178.3071196} proposes a genetic-programming based method to convert a regex into one with more efficient matching.  However, their method only supports pure regexes and does not guarantee invulnerability.

While our work concerns {\em repairing} vulnerability, there has been considerable work on the related problem of {\em detecting} vulnerability~\cite{10.1007/978-3-662-54580-5_1, Shen:2018:RCR:3238147.3238159,10.1007/978-3-319-40946-7_27, 10.1007/978-3-642-38631-2_11, SatoshiSugiyama2014}.  It is worth noting that while some (namely \cite{10.1007/978-3-319-40946-7_27,SatoshiSugiyama2014,10.1007/978-3-662-54580-5_1}) proposes to detect vulnerability formally rather than experimentally, no prior work on formal vulnerability detection supports the real-world extensions.  Whether a sound-and-complete detection of vulnerability for real-world regexes is possible is an open question.

Another related work is a recent work by Davis et al.~\cite{DavisSP21} that proposes a regex engine optimization to eliminate super-linear behavior of real-world regex matching at run time.
Finally, a recent work by Loring et al.~\cite{10.1145/3314221.3314645} presents a dynamic symbolic execution method for real-world regexes that addresses the regex-engine-dependent capturing issue mentioned in Section~\ref{sec:limit}.


\section{Conclusion}

We have presented a novel PBE regex repair method that guarantees the invulnerability of synthesized regexes and supports real-world regexes containing extended features of lookarounds, capturing groups, and backreferences.  For this, we have defined
a novel formal semantics of backtracking matching algorithm for real-world regexes and a formal definition of its time complexity.
With them, we have defined the {\em first formal definition of ReDoS vulnerability for real-world regexes}.  Additionally, we have presented a novel condition called {\em real-world strong 1-unambiguity} (\ltp{}) which we proved to be sound for guaranteeing ReDoS invulnerability of real-world regexes, formalized the \ltp{} repair problem and proved its NP-hardness.  We have presented an algorithm for solving the \ltp{} repair problem and experimentally evaluated its implementation, \tool{}, on a real-world data set.  The evaluation have shown that \tool{} can repair vulnerable real-world regexes successfully and efficiently.

To the best of our knowledge, we are the first to tackle the ReDoS vulnerabilities for real-world regexes and the challenge of repairing them, whose theoretical properties are substantially different from that of pure regexes which are tackled by prior works~\cite{Alquezar94incrementalgrammatical,10.1145/3093335.2993244, 6994453, 7374717, 10.1145/3360565, DBLP:journals/corr/abs-1908-03316, FlashRegex,10.1145/3129416.3129440,10.1145/3071178.3071196} that only considered pure regexes and/or did not concern ReDoS vulnerability.

\section*{Acknowledgements}
We thank the anonymous reviewers for their useful comments.
This work was supported by JSPS KAKENHI Grant Numbers 17H01720, 18K19787, 	20H04162, and 20K20625.

\bibliographystyle{IEEEtran}
\bibliography{main}


\appendices

\section{Full Rules of the Formal Semantics}
\begin{figure*}[t]\footnotesize
\begin{tabular}{cc}
    \begin{minipage}{.5\linewidth}
\infrule[Set of characters]
{ p < |w| \andalso w[p] \in C }
{([C], w, p, \Gamma)  \leadsto{}\ \{ (p+1, \Gamma) \} }
\andalso

\infrule[Set of characters Failure]
{ p \geq |w| \vee w[p] \notin C }
{([C], w, p, \Gamma)  \leadsto{}\ \emptyset }
\andalso

\infrule[Empty String]
{ }
{(\epsilon, w, p, \Gamma)  \leadsto{}\ \{ (p, \Gamma) \} }
\andalso

\infrule[Concatenation]
{(r_1, w, p, \Gamma) \leadsto{} \mathcal{N} \andalso \forall (p_i, \Gamma_i) \in \mathcal{N},\ (r_2, w, p_i, \Gamma_i) \leadsto{} \mathcal{N}_i}
{(r_1 r_2, w, p, \Gamma) \leadsto{} \bigcup_{ 0 \leq i < |\mathcal{N}| } \mathcal{N}_i }
\andalso

\infrule[Union]
{ (r_1, w, p, \Gamma) \leadsto{} \mathcal{N} \andalso (r_2, w, p, \Gamma) \leadsto{} \mathcal{N}' }
{(r_1 | r_2, w, p, \Gamma) \leadsto{} \mathcal{N} \cup \mathcal{N}'}
\andalso



\infrule[Repetition]
{ (r, w, p, \Gamma) \leadsto{} \mathcal{N} \\\andalso \forall (p_i, \Gamma_i) \in (\mathcal{N} \backslash \{ (p,\Gamma) \} ), \ (r^*, w, p_i, \Gamma_i) \leadsto{} \mathcal{N}_i }
{(r^*, w, p, \Gamma) \leadsto{} \{ (p, \Gamma) \} \cup \bigcup_{ 0 \leq i < |(\mathcal{N}\backslash \{(p,\Gamma)\})|} \mathcal{N}_i }

\end{minipage}
\begin{minipage}{.5\linewidth}

\infrule[Capturing group]
{(r, w, p, \Gamma) \leadsto{} \mathcal{N}}
{((r)_j, w, p, \Gamma) \leadsto{} \{ (p_i, \Gamma_i[ j \mapsto w[p..p_i) ]) \mid (p_i,\Gamma_i) \in \mathcal{N} \} }
\andalso 

\infrule[Backreference]
{\Gamma(i) \neq \bot \andalso (\Gamma(i),w,p,\Gamma) \leadsto{} \mathcal{N} }
{(\backslash i, w, p, \Gamma) \leadsto{} \mathcal{N} }
\andalso

\infrule[Backreference Failure]
{ \Gamma(i) = \bot }
{(\backslash i, w, p, \Gamma) \leadsto{} \emptyset}
\andalso


\infrule[Positive lookahead]
{(r, w, p, \Gamma) \leadsto{} \mathcal{N}  }
{(\text{(?=}r\text{)}, w, p, \Gamma) \leadsto{} \{ (p,\Gamma') \mid (\_,\Gamma') \in \mathcal{N} \} }
\andalso

\infrule[Negative lookahead]
{(r, w, p, \Gamma) \leadsto{} \mathcal{N} \andalso{} \mathcal{N}' = \mathit{ite}(\mathcal{N} \neq \emptyset, \emptyset, \{(p,\Gamma)\})}
{(\text{(?!}r\text{)}, w, p, \Gamma) \leadsto{} \mathcal{N}' }
\andalso

\infrule[Positive lookbehind]
{(x, w[p-|x|..p), 0, \Gamma) \leadsto{} \mathcal{N} \andalso \mathcal{N}' = \mathit{ite}(\mathcal{N} \neq \emptyset, \{(p,\Gamma)\}, \emptyset)}
{(\text{(?\textless=}x\text{)}, w, p, \Gamma) \leadsto{} \mathcal{N}' }
\andalso

\infrule[Negative lookbehind]
{(x, w[p-|x|..p), 0, \Gamma) \leadsto{} \mathcal{N} \andalso \mathcal{N}' = \mathit{ite}(\mathcal{N} \neq \emptyset, \emptyset, \{(p,\Gamma)\})}
{(\text{(?\textless!}x\text{)}, w, p, \Gamma) \leadsto{} \mathcal{N}' }

\end{minipage}
\end{tabular}
\caption{Rules of the matching relation $\leadsto$}
\label{fig:fullsemanticsv} 
\end{figure*}

The full rules for deriving the matching relation $\leadsto$ is shown in Figure~\ref{fig:fullsemanticsv}.
We describe the rules for the pure regex features which were not explained in Section~\ref{sec:regex}.  In the two rules for a set of characters, the regex $[C]$ tries to match the string $w$ at the position $p$ with the function capturing $\Gamma$. If the $p$-th character $w[p]$ is in the set of character $C$, then the matching succeeds returning the matching result $(p+1, \Gamma)$ (\textsc{Set of characters}).  Otherwise, the character $w[p]$ does not match or the position is at the end of the string, and $\emptyset$ is returned as the matching result indicating the match failure (\textsc{Set of characters Failure}).
%
The rules (\textsc{Empty String}), (\textsc{Concatenation}), (\textsc{Union}) and (\textsc{Repetition}) are self explanatory.  Note that we avoid self looping in (\textsc{Repetition}) by not repeating the match from the same position.

\section{Full rules for generating consistency-with-example constraints}
\begin{figure*}[t]\footnotesize
\begin{tabular}{cc}
    \begin{minipage}{.5\linewidth}

\infrule[Set of characters]
{ p < |w| \andalso{} w[p] \in C }
{ ([C], w, p, \Gamma, \phi) \encode{}  (\{ (p+1,\Gamma, \phi) \}, \emptyset) }
\andalso

\infrule[Concatenation]
{(\sstate_1, w, p, \Gamma, \phi) \encode{} (\mathcal{S}, \mathcal{F}) \andalso\\ \forall (p_i, \Gamma_i, \phi_{i}) \in \mathcal{S}.\ (\sstate_2, w, p_i, \Gamma_i, \phi_{i}) \encode{} (\mathcal{S}_i, \mathcal{F}_i)}
{(\sstate_1 \sstate_2, w, p, \Gamma, \phi) \encode{} (\bigcup_{0 \leq i < |\mathcal{S}|}\mathcal{S}_i, \mathcal{F} \cup \bigcup_{0 \leq i < |\mathcal{S}|} \mathcal{F}_i) }
\andalso

\infrule[Union]
{ (\sstate_1, w, p, \Gamma, \phi) \encode{} (\mathcal{S}_1, \mathcal{F}_1) \andalso (\sstate_2, w, p, \Gamma, \phi) \encode{} (\mathcal{S}_2, \mathcal{F}_2) }
{(\sstate_1 | \sstate_2, w, p, \Gamma, \phi) \encode{} (\mathcal{S}_1 \cup \mathcal{S}_2, \mathcal{F}_1 \cup \mathcal{F}_2) }
\andalso

\infrule[Repetition]
{ (\sstate, w, p, \Gamma, \phi) \encode{} (\mathcal{S}, \mathcal{F}) \andalso\\ \forall (p_i, \Gamma_i, \phi_{i}) \in (\mathcal{S} \backslash \{ (p,\Gamma, \_) \} ).\ (t^*, w, p_i, \Gamma_i, \phi_{i}) \encode{} (\mathcal{S}_i, \mathcal{F}_i) }
{(\sstate^*, w, p, \Gamma, \phi) \encode{} ( \{ (p, \Gamma, \phi) \} \cup \bigcup_{0 \leq i < |\mathcal{S}|} \mathcal{S}_i, \emptyset ) }
\andalso

\infrule[Capturing group]
{(\sstate, w, p, \Gamma, \phi) \encode{} (\mathcal{S}, \mathcal{F})}
{((\sstate)_i, w, p, \Gamma, \phi) \encode{} (\bigcup_{ (p_i, \Gamma_i, \phi_{ci}) \in \mathcal{S} } (p_i, \Gamma_i[ i \mapsto w[p..p_i) ], \phi_{ci}), \mathcal{F}) }
\andalso

\infrule[Backreference]
{ \text{Let $x$ = }\Gamma(i) \andalso x = w[p..p+|x|) }
{(\backslash i, w, p, \Gamma, \phi) \encode{} ( \{ (p+|x|, \Gamma, \phi) \} , \emptyset) }

\end{minipage}
\begin{minipage}{.5\linewidth}


\infrule[Positive lookahead]
{(\sstate, w, p, \Gamma, \phi) \encode{} (\mathcal{S}, \mathcal{F}) }
{(\text{(?=}\sstate\text{)}, w, p, \Gamma, \phi) \encode{} ( \{(p,\Gamma',\phi') \mid (\_,\Gamma', \phi') \in \mathcal{S}\}, \mathcal{F}) }
\andalso

\infrule[Negative lookahead]
{(\sstate, w, p, \Gamma, \phi) \encode{} (\mathcal{S}, \mathcal{F})}
{(\text{(?!}\sstate\text{)}, w, p, \Gamma, \phi) \encode{}\\ ( \{ (p,\Gamma,\phi') \mid (\bot,\bot,\phi') \in \mathcal{F} \}, \{ (\bot,\bot,\phi') \mid (\_,\_,\phi') \in \mathcal{S} \} ) } 
\andalso

\infrule[Positive lookbehind]
{(x, w[p-|x|,p), 0, \Gamma, \phi) \encode{} (\mathcal{S}, \mathcal{F})}
{(\text{(?\textless=}x\text{)}, w, p, \Gamma, \phi) \encode{} ( \{  (p,\Gamma, \phi') \mid (p', \Gamma', \phi')  \in \mathcal{S} \}, \mathcal{F}) }
\andalso

\infrule[Negative lookbehind]
{(x, w[p-|x|,p), 0, \Gamma, \phi) \encode (\mathcal{S}, \mathcal{F}) }
{(\text{(?\textless!}x\text{)}, w, p, \Gamma, \phi) \encode{}\\ (\{ (p, \Gamma, \phi') \mid (\bot, \bot, \phi') \in \mathcal{F} \},  \{ (\bot,\bot,\phi') \mid (\_,\_,\phi') \in \mathcal{S} \} ) }
\andalso

\infrule[Hole]
{\text{$\hole$ is the $i$-th hole}}
{(\hole, w, p, \Gamma, \phi) \encode (\{ (p+1, \Gamma, \phi \land v_i^{w[p]}) \}, \{ (\bot, \bot, \phi \land \lnot v_i^{w[p]}) \}) }

\end{minipage}
\end{tabular}
\caption{Rules for generating consistency-with-examples constraints.}
\label{tab:fullencode} 
\end{figure*}

The full rules for generating the consistency-with-examples constraints is shown in Figure~\ref{tab:fullencode}.
The cases where the matching fails, that is, $(r, w, p, \Gamma, \phi) \encode (\emptyset, \{(\bot,\bot,\phi)\})$, are omitted.

\section{The proof of Theorem~\ref{theo:nphardness}}

We first review \textsc{ExactCover}.
\begin{definition}[Exact Cover]
\normalfont
Given a finite set $\mathcal{U}$ and $\mathcal{S} \subset \mathcal{P}(\mathcal{U})$,
\textsc{ExactCover} is the problem of deciding if there exists $\mathcal{S'} \subseteq \mathcal{S}$ such that for every $i \in \mathcal{U}$, there is a unique $S \in \mathcal{S'}$ such that $i \in S$.
\end{definition}

\begin{proof}
We give a reduction from the exact cover to the repair problem.  Let  $\mathcal{S} = \{ S_1, S_2, ..., S_k \}$.
We create the following (decision version of) \ltp{} repair problem:
\begin{itemize}
\item The alphabet $\Sigma = \mathcal{U}$;
\item The set of positive examples $P = \mathcal{U}$;
\item The set of negative examples $N = \emptyset$;
\item The distance bound is $2k$; and
\item The pre-repair expression $r_1 = r_{11}r_{12}$ where $r_{11}$ and $r_{12}$ are as defined below:
\[
\setlength\arraycolsep{2pt}
\begin{array}{rcl}
r_{11} & = & \epsilon (\mbox{?=}[S_1])^{2k} (\epsilon)_1 [S_1] (\epsilon)_2 \\
       & \mid & \epsilon (\mbox{?=}[S_2])^{2k} (\epsilon)_3 [S_2] (\epsilon)_4 \\ 
       & \mid & ... \\
       & \mid & \epsilon (\mbox{?=}[S_k])^{2k} (\epsilon)_{2k-1} [S_k] (\epsilon)_{2k} \\
r_{12} & = & ((?!\backslash1)|(\mbox{\mbox{?=}}\backslash1\backslash2))^{2k}\\&&((?!\backslash3)|(\mbox{?=}\backslash3\backslash4))^{2k} ... ((?!\backslash2k-1)|(\mbox{?=}\backslash2k-1\backslash2k))^{2k}.
\end{array}
\]
\end{itemize}
Here, $r^{2k}$ is the expression obtained by concatenating $r$ $2k$ times.

It is easy to see that this is a polynomial time reduction since the construction of $r_1$ can be done in time cubic in the size of the input \textsc{ExactCover} instance.
Also, note that the above is a valid \ltp{} repair problem instance because $P = \mathcal{U} \subseteq L(r_1)$ and $L(r_1) \cap N = \emptyset$.
We show that reduction is correct, that is, the input \textsc{ExactCover} instance has a solution iff there exists $r_2$ satisfying conditions (1)-(3) of Definition~\ref{def:ltprepair} and $\distfunc{}(r_1, r_2) \leq 2k$.  First, we show the only-if direction, let $\mathcal{S'} \subset \mathcal{S}$ be a solution to the \textsc{ExactCover} instance.  The repaired expression $r_2 = r_{21} r_{22}$ where $r_{22} = r_{12}$, and 
$r_{21}$ is $r_{11}$ but with each $i$-th head $\epsilon$ in the union replaced by $[\emptyset]$ iff 
$S_i \notin \mathcal{S'}$.  Note that $\distfunc{}(r_1, r_2) = 2|\mathcal{S}\setminus\mathcal{S'}| \leq 2k$.  Also, $r_2$ satisfies the \ltp{} condition because for every $a \in \mathcal{U}$, there exists only one $S_i \in \mathcal{S'}$ such that $a \in S_i$, i.e., on any input string starting with $a$, we deterministically move to the $i$-th choice in the union (and there are no branches after that point).  Also, $r_2$ correctly classifies the examples.  To see this, consider an arbitrary $a \in P = \mathcal{U}$.  Then, $a$ is included in some $S_i \in \mathcal{S'}$.  Therefore, the matching passes the $r_{21}$ part with successful captures at indexes $2i-1$ and $2i$, and passes the $r_{22}$ part because the negative lookahead $(?!\backslash j)$ succeeds for all $j \neq i$ and the positive lookahead $(\mbox{?=}\backslash2i-1\backslash2i)$ succeeds.  Thus, $r_2$ is a correct repair.

We show the if direction.  First, note that any valid repair of $r_1$ must preserve the $k$ union choices of $r_{11}$ because deleting any union choice would already exceed the cost $2k$.  From this, it is not hard to see that the only possible change is to change the head $\epsilon$ in the union choices in $r_{11}$.  For instance, it is useless to change $[S_i]$ to some $r$ where $L(r)$ contains elements not in $S_i$ because of the $2k$ many $(\mbox{?=}[S_i])$ preceding it.  Note that changing $(\mbox{?=}[S_i])^{2k}$ would exceed the cost.  Nor, can $[S_i]$ be changed to some $r$ where $L(r)$ does not contain an element of $S_i$ because of the capturing group $(\epsilon)_{2i}$ and $(\epsilon)_{2i-1}$ before and after $[S_i]$ and the check done in $r_{12}$.  Note that changing any of the check in $r_{12}$ would again exceed the cost.  This also shows that the capturing groups $(\epsilon)_{2i}$ and $(\epsilon)_{2i-1}$ cannot be changed. Therefore, the only meaningful change that can be done is to change some of the head $\epsilon$ in $r_{11}$ to some $r$.  Note that for any $r$ chosen here, by the \ltp{} property, $r_2$ will not accept $\{ a \mid aw \in L(r) \}$ as any input $a \in \{ a \mid aw \in L(r) \}$ would direct the match algorithm deterministically to this choice but the match would fail when it proceeds to $[S_i]$.  Therefore, the only change that can be done is to change it to some $r$ such that $L(r) = \emptyset$.  Then, from a successful repair $r_2$, we obtain the solution $\mathcal{S'}$ to the \textsc{ExactCover} instance  where $S_i \notin \mathcal{S'}$ iff the $i$-th head $\epsilon$ in $r_{21}$ is changed to some $r$ such that $L(r) = \emptyset$.
\end{proof}

\section{Correctness of \ltp{}}
\label{subsec:proof}
In this section, we show that a regex that satisfies \ltp{} is invulnerable.
Before we go on with the main proof, we show that lookaheads that satisfy \ltp{} runs in constant time to eliminate lookaheads from the later arguments.
\begin{theorem}
\label{theo:starfree}
A regex that does not contain repetitions, unions, and backreferences runs in constant time.
\end{theorem}
\begin{proof}
We prove that, for such a regex $r$, the size of $\mathcal{N}$ where $(r,w,0,\emptyset)\leadsto \mathcal{N}$ is constant.
This proof is by induction on the structure of the regex. 
\end{proof}
By the definition of \ltp{}, lookaheads that satisfy \ltp{} do not contain repetitions and backreferences.
Hence, lookaheads in a regex that satisfies \ltp{} also run in constant time.
For this, lookarounds, lookbehinds, and empty strings run in constant time.  Also, they consume no characters.
Thus, in what follows, without loss of generality, we assume that a regex does not contain empty strings, lookaheads, and lookbehinds.

We map a derivation tree to a {\em directed tree}.
\begin{definition}[Directed Tree]
\normalfont
A {\em directed graph} is a tuple $G = (V,E)$.
Here, $V$ is a finite set of vertices and $E$ is a finite set of directed edges.
A vertex $v = (\freshi{}, p, \Gamma) \in V$ consists of a unique index $\freshi{}$ and a matching result $(p, \Gamma)$.
A directed edge ({\em edge} for short) $e = (v_1, v_2) \in E$ consists of two vertices $v_1$ (often called {\em tail}) and $v_2$ (often called {\em head}).  A {\em directed tree} is a directed graph that is of a tree shape (i.e., has no cycles and $|E| = |V|-1$.).

\end{definition}
We define the size of a directed tree $G = (V,E)$ as the size of $E$, i.e., $|G| = |E|$.
We use the notation $\indig(v) = |\{ v' \mid (v', v) \in E \}|$, $\outdig(v) = |\{ v' \mid (v, v') \in E \}|$, \leaf{}$(G) = \{ v \mid v \in V \land \outdig(v) = 0 \}$, \groot{}$(G) = v$ such that $v \in V \land \indig(v) = 0$, $\tail{}(e) = v$ and $\head{}(e) =v'$ for an edge $e=(v,v')$, and $E(v)$, where $v \in V$, for a set of heads, i.e., $E(v) = \{ v' \mid (v,v') \in E \}$.
For a tuple $t = (t_1,\cdots,t_n)$, we write $\#_i(t)$ for $t_i$, where $1 \leq i \leq n$.

We define a construction $\togform{}$ from a derivation tree $A$ to the directed tree $G$ as follows: 
Here, $\fresh{}()$ returns a fresh identifier, and, $a := b$ means that $a$ is replaced with $b$.  For a directed tree $G = (V,E)$ and $G' = (V',E')$ such that $V \cap V' = \emptyset$ and $v \in \leaf(G)$, we write $G[v \mapsto G']$ for the graph $(V \cup V'\setminus \{v\}, E \cup E' \cup \{(v',\groot(G'))\} \setminus \{ (v',v) \})$ where $v'$ is the unique vertex such that $(v',v) \in E$.  I.e., $G[v \mapsto G']$ is the graph obtained by replacing the leaf $v$ of $G$ by (the root of) $G'$.
Additionally, we assume that $G_i = (V_i, E_i)$.
\begin{itemize}
    \item Case (\textsc{Set of characters}). $G = (\{v_1, v_2\},\{ (v_1,v_2) \})$, where $v_1 = (\fresh(), p, \Gamma)$ and $v_2 = (\fresh(), p+1, \Gamma)$.
    \item Case (\textsc{Set of characters Failure}). $G = (\{v_1, v_2\},\{ (v_1,v_2) \})$, where $v_1 = (\fresh(), p, \Gamma)$ and $v_2 = (\fresh(), \emptyset{}, \Gamma)$.    
    \item Case (\textsc{Concatenation}). Let $G_1 = \togform{}((r_1, w, p, \Gamma) \leadsto{} \mathcal{N})$. For all $(p_i, \Gamma_i) \in \mathcal{N}$, there exists a vertex $(\_, p_i, \Gamma_i) \in \leaf{}(V)$.
    For all $(\_, p_i, \Gamma_i) \in \leaf{}(V)$, let $G_{2i} = \togform{}((r_2,w,p_i,\Gamma_i) \leadsto{} \mathcal{N}_i)$, $G_1 := G_1[(\_, p_i, \Gamma_i) \mapsto{} G_{2i}]$.   $G = (V_1 \cup \{ v \}, E_1 \cup \{ (v,\groot(G_1)) \}$ where $v = (\fresh(),p,\Gamma)$.
    
    \item Case (\textsc{Union}). Let $G_1 = \togform{}((r_1, w, p, \Gamma) \leadsto{} \mathcal{N})$ and $G_2 = \togform{}((r_2, w, p, \Gamma) \leadsto{} \mathcal{N'})$. $G = ( V_1 \cup V_2 \cup \{ v \}, E_1 \cup E_2 \cup \{ (v, \groot{}(G_1)), (v, \groot{}(G_2)) \} )$, where $v = (\fresh(), p, \Gamma)$.
    
    \item Case (\textsc{Repetition}). 
    Let $G_1 = \togform{}((r,w,p,\Gamma)\leadsto{}\mathcal{N})$.
    For all $(p_i, \Gamma_i) \in \mathcal{N}$, there exists a vertex $(\_, p_i, \Gamma_i) \in \leaf{}(V_1)$.
    For all $(\_, p_i, \Gamma_i) \in leaf{}(V_1)$, let $G_{2i} = \togform{}((r^*,w,p_i,\Gamma_i)\leadsto{}\mathcal{N}_i)$, $G_1 := G_1[(\_, p_i, \Gamma_i) \mapsto{} G_{2i}]$.
    $G = (V_1 \cup \{ v_1, v_2 \}, E_1 \cup \{ (v_1,v_2), (v_1,\groot{}(G_1)) \} )$, where $v_1 = (\fresh(), p, \Gamma)$ and $v_2 = (\fresh(), p, \Gamma)$.
    
    \item Case (\textsc{Capturing group}).
    Let $G_1 = \togform{}((r,w,p,\Gamma) \leadsto{} \mathcal{N})$.
    $G = ( V_1 \cup \{ v \}, E_1 \cup \{ (v, \groot{}(G_1)) \} )$.
    
    \item Case (\textsc{Backreference}).
    Let $G_1 = \togform{}((\Gamma(i),w,p,\Gamma)\leadsto{}\mathcal{N})$.
    $G = (V_1 \cup \{ v \} , E_1 \cup \{ (v,\groot{}(G_1)) \})$ where $v = (\fresh(),p,\Gamma)$.
    \item Case (\textsc{Backreference Failure}). $G = (\{v_1, v_2\},\{ (v_1,v_2) \})$, where $v_1 = (\fresh(), p, \Gamma)$ and $v_2 = (\fresh(), \emptyset{}, \Gamma)$.    
\end{itemize}

\begin{lemma}
\label{lem:directediseq}
Given a derivation tree $(r,w,0,\emptyset{}) \leadsto{} \mathcal{N}$.
Let $A$ be the derivation tree and $G = \togform{}(A)$ be the directed tree.
The size of the derivation tree $A$ is greater than or equal to the size of the directed tree $G$.
\end{lemma}
\begin{proof}
The proof is by induction on the structure of $A$.
\end{proof}


\begin{definition}[Main and Sub Branch]
\normalfont
Let $G = (V,E)$ be a directed tree.
For each vertex $v \in V$, we say an edge $e \in E(v)$ is a {\em main branch} of $v$ if $\forall  e' \in E(v) \setminus \{ e \}, \#_2(\head{}(e')) < \#_2(\head{}(e))$.
Otherwise, we say the edge is {\em sub branch} of $v$. 
\end{definition}

\begin{definition}[Main Path]
\normalfont
We say a sequence of main branches $\path{}_m = e_1 e_2 \cdots e_n$ is a {\em main path} if $\head{}(e_i) = \tail{}(e_{i+1})$ for $1 \leq i < n$, $\tail{}(e_1)$ is a root, i.e., $\indig{}(\tail{}(e_1)) = 0$, and, for every $e \in E(\head{}(e_{n}))$, $e$ is a sub branch. 
\end{definition}

By the definition of the main path, there is at most one main path in a directed tree.
\begin{lemma}
\label{lem:mainisconstant}
Given a directed tree $G = (V,E)$, which is obtained by $\togform{}((r,w,0,\emptyset) \leadsto{} \mathcal{N})$.
If $G$ has a main path $\path{}_m$, then the length $|\path{}_m|$ is $O(|w|)$.
\end{lemma}
\begin{proof}
By induction on the structure of $r$.
The only interesting case is when $r$ is a repetition, say, $r = {r'}^*$.
We show that ${r'}^*$ consumes at most $O(|w|)$ characters during the whole matching.
Let $n$ be the number of iterations of ${r'}^*$ on $w[p_1..|w|)$ and $({r'}^*, w, p_i, \Gamma_i) \leadsto{} \mathcal{N}_i$ be the $i$-th iteration, where $1 \leq i \leq n$.
Let $e_1 e_2 \cdots e_n$ be the main path.
For $1 \leq i < n$, $\#_2(\tail{}(e_i)) <\#_2(\head{}(e_{i}))$ because if $\tail{}(e_i) = \head{}(e_{i})$, then it means that $r'$ accepts an empty string and so it violates \ltp{} because there are two or more paths to the first alphabet in $r'$ or the next expressions.
Hence, $\#_2(\tail{}(e_i)) < \#_2(\head{}(e_{i}))$ and ${r'}^*$ consumes at most $O(|w|)$ characters.
\end{proof}

\begin{definition}[$\epsilon$ Subtree]
\normalfont
Let $G = (V,E)$.
We say a subtree $G_{\epsilon} = (V_{\epsilon}, E_{\epsilon})$, where $V_{\epsilon} \subseteq V$ and $E_{\epsilon} \subseteq E$, is an {\em $\epsilon$ subtree} if every $e \in E_{\epsilon}$ is a sub branch, 
$|E_\epsilon{}(\groot{}(G_\epsilon))| = 1$, 
$\#_2(\tail{}(e)) = \#_2(\head{}(e))$ or $\#_2(\head{}(e)) = \emptyset{}$ for every $e \in \{ (v,v') \in E_{\epsilon} \mid v \neq \groot(G_\epsilon) \}$, and, for every $v_1, v_2 \in V_{\epsilon}$, where $v_1 \not= v_2$, there exists a sequence of $ E_{\epsilon}$ edges $e_1 e_2 \cdots e_n$ such that $\tail{}(e_1) = v_1$, $\head{}(e_n) = v_2$, $\head{}(e_i) = \tail{}(e_{i+1})$ for $1 \leq i < n$.
\end{definition}

\begin{lemma}
\label{lem:esubtreeisconstant}
Given an $\epsilon$ subtree $G_\epsilon$ in $\togform{}((r,w,p,\Gamma) \leadsto{} \mathcal{N})$.
The size of the $\epsilon$ subtree $|G_\epsilon|$ is constant.
\end{lemma}
\begin{proof}
Suppose that $|G_\epsilon|$ is not constant. 
Then, $r$ contains a repetition ${r'}^*$ and the repetition ${r'}^*$ iterates at least twice because, if not, then the size is $O(|r|)$, i.e., constant.
By the definition of $\epsilon$ subtrees, the repetition ${r'}^*$ does not consume any character and this means that ${r'}$ accepts an empty string.
However, it means that the repetition ${r'}^*$ violates the \ltp{} condition because there are two or more paths to the first character in ${r'}$ or the next expression.
Thus, ${r'}^*$ iterates at most once.
Hence, $|G_\epsilon|$ is constant.
\end{proof}

\begin{lemma}
\label{lem:numofesubtreeisw}
Let $G = \togform{}((r,w,p,\Gamma)\leadsto{}\mathcal{N})$.
Then the number of $\epsilon$ subtrees in $G$ is $O(|w|)$.
\end{lemma}
\begin{proof}
We show that the number of $\epsilon$ subtrees is $O(|\path{}_m|)$, i.e., $O(|w|)$.
The proof is by induction on structure of derivation trees.
Here, we only focuses on the case (\textsc{Repetition}).
In the case of (\textsc{Repetition}), 
let $m$ be the number of iterations of $({r'}^*,w,p',\Gamma') \leadsto{} \mathcal{N}'$.
For the $i$-th iteration, let $(r', w, p_i, \Gamma_i) \leadsto{} \mathcal{N}_i'$, where $1 \leq i \leq m$, $p_1 = p'$, and $\Gamma_1 = \Gamma'$.
Then, by inductive hypothesis, each derivation tree $(r', w, p_i, \Gamma_i) \leadsto{} \mathcal{N}_i'$ satisfies the assertion.
Hence, the assertion holds.
\end{proof}


\begin{lemma}
\label{lem:belongstomore}
Let $G = (V,E) = \togform{}((r, w, 0, \emptyset{}) \leadsto{} \mathcal{N})$.
For all edges $e \in E$, $e$ belongs to either a main path or an $\epsilon$ subtree.
\end{lemma}
\begin{proof}
The proof is by induction on the structure of derivation trees.
\end{proof}

\begin{theorem}
\label{theo:main}
Given a directed tree $G = (V,E)$.
The size of the directed tree $|G|$ is $O(|w|)$.
\end{theorem}
\begin{proof}
By Lemma \ref{lem:belongstomore}, $G$ consists of a main path and $\epsilon$ subtrees.
By Lemma \ref{lem:mainisconstant}, the size of the main path is $O(|w|)$.
By Lemma \ref{lem:esubtreeisconstant}, the size of the $\epsilon$ subtrees is $O(1)$, and by Lemma \ref{lem:numofesubtreeisw}, the number of $\epsilon$ subtrees is $O(|w|)$.
Hence the size of $G$ is $O(|w| \times 1 + 1 \times |w|) = O(|w|)$.
\end{proof}

Finally, we are now ready to proof Theorem \ref{theo:numstlin}.

\begin{proof}{(Proof of Theorem \ref{theo:numstlin})}
Immediate from Lemma \ref{lem:directediseq} and Theorem \ref{theo:main}.
\end{proof}

\section{Insufficiency of Deterministic Regexes}
\label{appendix:insufficiency_of_flashregex}
The details of why \cite{FlashRegex} is insufficient for guaranteeing unambiguity is as follows. The idea of \cite{FlashRegex} for repairing ReDoS-vulnerability is to synthesize so-called ``deterministic'' (also called 1-unambiguous \cite{BRUGGEMANNKLEIN1998182,10.1007/3-540-57273-2_45}) regexes. The definition of deterministic regex is as follows:
\begin{definition}[Deterministic, Definition 2.1 in \cite{BRUGGEMANNKLEIN1998182,10.1007/3-540-57273-2_45}]
\normalfont
A regex $E$ is {\em deterministic} (or {\em 1-unambiguous}) iff, for all words $u,v,w \in \Pi^*$ and all symbols $x,y \in \Pi$, \[
uxv, uyw \in L(E') \land x \not= y \Rightarrow x^\natural \not= y^\natural.
\]
\end{definition}
Here, $\Pi$ is the subscripted alphabet $\{ a_i \mid a \in \Sigma, i \in \mathbb{N}\}$, $x^\natural$ is the character obtained by dropping the subscript of $x \in \Pi$ (e.g., $(a_1)^\natural = a$), and $E'$ is $E$ but with each characters in $E$ annotated with distinct subscripts (e.g., if $E = ab(a|b)c$, then $E$’$ = a_1 b_1 (a_2 | b_2) c_1$). 

Intuitively, a regex $E$ \tchanged{}{is} deterministic iff for any $w \in L(E)$, there is a unique subscripted word $v \in L(E$’$)$ such that $v^\natural = w$. Thus, the vulnerable regex $E = (a^*)^*$ satisfies the definition because its subscripted regex $E$' is $(a_1^*)^*$ and $L(E$’$) = \{\epsilon, a_1, a_1 a_1, ...\}$ (i.e., there are no $uxv, uyw \in L(E$’$)$ such that $x \neq y$). Hence, $(a^*)^*$ is a deterministic regex while it is vulnerable as we have shown in Section \ref{subsec:formalsemantics}. Consequently, synthesizing deterministic regexes is insufficient for guaranteeing ReDoS invulnerability.


\end{document}